\documentclass[10pt]{article}
\usepackage{epsfig,graphics,graphicx,subfigure,color}
\usepackage{amsfonts,amssymb,amsthm,amsmath,amsbsy,mathtools,cuted,kotex,subeqnarray,ocgx}
\newcommand{\interior}[1]{  {\kern0pt#1}^{\mathrm{o}}}
\newcommand{\norm}[1]{\left\lVert#1\right\rVert}

\newtheorem{theorem}{Theorem}
\newtheorem{corollary}[theorem]{Corollary}
\newtheorem{lem}[theorem]{Lemma}

\def\Bbb R{{\rm \bf R}}

\def\gathered{\begin{array}{c}}
\def\endgathered{\end{array}}
\def\text{\mbox}

\begin{document}
\title  {The number of global solutions for GPS source localization in two-dimension }
\author{Kiwoon KWON\footnote{e-mail address:kwkwon@dongguk.edu, Dept. of Math., Dongguk University, 04620 Seoul, South Korea.}
}
\maketitle

\begin{abstract}
Source localization is widely used in many areas including GPS, but the influence of possible noises is not so negligible.
Many optimization methods are attempted to alleviate different kinds of noises.  Needless to say the stability of the solution,
even  the number of global solutions are not fully known. Only local convergence or stability for the optimization problem are known in simple $L^1$\cite{Kwon} or $L^2$\cite{Kwon3} settings. 
In this paper, we prove that the number of possible two dimensional source locations with three measurements  in $L^2$ setting, is at most $5$, which is the complement and correction to 
the previous work \cite{Kwon3}. We also showed the sufficient and necessary condition for the number of the solutions being 1,2,3,4,and 5, where the measurement triangle 
is isosceles and  the measurement distance for the two isosceles triangle bases are the same.

\end{abstract}

\section{Introduction}
Let $Z_1, Z_2, Z_3\in \mathbb R^2$ be the sensor locations and $X\in \mathbb R^2$ be the possible source position. 
Denote $d_i, i=1,2,3$ as the measurement length between the $i$-th sensor and the source having noise $\epsilon_i$:
\begin{equation}\label{eq:main}
                         d_i = ||X-Z_i||+\epsilon_i.
 \end{equation}
 Let $Z=(Z_1,Z_2,Z_3) \in (\mathbb R^2)^3$, $d=(d_1,d_2,d_3)\in\mathbb R^3 , \epsilon =(\epsilon_1,\epsilon_2,\epsilon_3)\in \mathbb R^3$, 
and $X$ be the solutions of the following minimization problem: 
 \begin{equation}\label{eq:main2}
                    X = \mathrm{argmin}_{W\in\mathbb R^2} O(W), \mbox{ where } O(W)= \sum_{j=1}^3 |||W-Z_j||^2-d_j^2|
 \end{equation}

By Proposition 1 \cite{Kwon}, there exists at least one solution $X$.  
If we notate $|X|$ as the number of possible locations set $X$, then $|X|\ge 1$.

Let $S_1,S_2,$ and $S_3$ be the measurement circles, and further define
$$
S=S_1\cup S_2 \cup S_3,\;\; S_{123} = S_{12}\cup S_{23}\cup S_{31}, \;\; S_{23}= S_2\cap S_3, \; S_{31} = S_3\cap S_1,\; S_{12}=S_1\cap S_2. 
$$
By definition, $S_{ij}=S_{ji}$. Assuming $|S_{12}|=2$, let $S_{12+}$ and $S_{12-}$ be the points having shortest and longest distances from $Z_3$ in $S_{12}$, respectively.  
If $|S_{12}|=1$, define $S_{12+}=S_{12-}=S_{12}$. We can also define $S_{23+}, S_{23-}, S_{31+},$ and $S_{31-}$ in the same way.  Denote $S_{123+}=\{S_{12+}, S_{23+}, S_{31+}\}$ and $S_{123-}=\{S_{12-}, S_{23-}, S_{31-}\}$.
 
Let $R_{i_1,i_2,i_3}$ be the intersection of closed disks with $i_l=1$ and the outside of open disks with $i_l=0$ for $l=1,2,$ and $3$. That is to say, $R_{i_1, i_2, i_3}$ is a closed set. 
Let 
$$ R_l = \cup_{i_1+i_2+i_3=l} R_{i_1, i_2,i_3}, $$ then we have
$$ \mathbb R^2 = R_0 \cup R_1 \cup R_2 \cup R_3. $$

The number of possible source locations to \eqref{eq:main2} is already explained in \cite{Kwon}. But Theorem 2(a),  2(b), and 6(c) in \cite{Kwon} are misleading. 
Therefore, we correct the theorems in Theorem \ref{th:main1} and \ref{th:main2}, where red phrases are the corrected ones to the thorems of \cite{Kwon}. 
Thus, Table 1 in \cite{Kwon} should be changed into Table \ref{tab:con}  in the present paper.
\begin{table}
\begin{centering}
\begin{tabular}{|c||c|c|}
\hline
                           &\cite{Kwon3}, the present work                                                      &\cite{Kwon}\\
\hline \hline
The objective function
                          &$O(W)=\sum_{j=1}^3 | d_j^2 - ||W-Z_j||^2|$                                     &$O(W)=\sum_{j=1}^3 |d_j -||W-Z_j|||$\\
\hline
Existence           &Holds                                                                                               &Holds\\
\hline
$|X|$                   &$1,2,3{\color{red},4,\mbox{ and }5}$                                              &$1,2,\cdots, \infty$ \\   
\hline
$|X|=1$ for nonempty $K$     
                          &Holds                                                                                                &Holds when $K$ is connected\\
\hline
Uniqueness when $K$ is empty
                         &Holds when $R_3=\phi$ or one of                                                    &Does not hold\\
                         &$R_{100},R_{010},$ and $R_{001}$ is empty or not connected.     &\\
\hline
Nonuniqueness examples
                        &Shown for all cases with $|X|=2,3{\color{red},4,\mbox{ and }5}$      &Shown for $|X|=2,3$\\   
                        &                                                                                                          &and not shown for $|X|=4,5,\cdots$\\                                     
\hline
Detailed position 
                        &Given {\color{red} only when $K$ is nonempty}                                &Given only when $K$ is nonempty\\
                        &{\color{red} or when  $|Z_3 Z_1|=|Z_3 Z_2|$ and $d_1=d_2$}         &\\ 
\hline                                             
\end{tabular}
\caption{Differences among the results in \cite{Kwon},\cite{Kwon3}, and the present work. $K=R_0 \cap \triangle Z_1 Z_2 Z_3$. }
\label{tab:con}
\end{centering}
\end{table}
Theorem 3 shows that the number of solutions can be $|X|=2,4,$ and $5$, where only the cases with $|X|=1,2,$ and $3$ are shown in \cite{Kwon}.

\begin{theorem}\label{th:main1}
The condition $|X|=2,3,{\color{red}4,\mbox{ and }5}$ holds if and only if $R_3\neq\phi$ and $R_{100},R_{010},$ and $R_{001}$ are connected and nonempty, respectively.
If one of these conditions holds, we have  $X=argmin_{A\in {\color{red} S_{123}^\pm}} \norm{A-Y_0}$.
\end{theorem}

From Theorem \ref{th:main1},  we have $|X|\le 6$. Thus, the following theorem holds if we prove $|X|\neq 6$:
\begin{theorem}\label{th:main2}
There are at most {\color{red} five} solutions.%
\end{theorem}

\begin{theorem}\label{th:main3}
 Suppose that 
 $$\norm{Z_1 Z_3} = \norm{Z_2 Z_3}, \; d_1 = d_2>\norm{Z_1 Z_3}, \; d_3^2 = d_1^2 - \norm{Z_1 Z_3}^2.$$
 Let us denote 
$$P =\sqrt{\frac{\norm{Z_1 Z_2}^2}{4} + \left( \norm{Z_1 Z_3}^2 - \frac{\norm{Z_1 Z_2}^2}4\right)  \left( \frac { \norm{Z_1 Z_3}^2 + \norm{Z_1 Z_2}^2} { \norm{Z_1 Z_3}^2 - \norm{Z_1 Z_2}^2}  \right)^2},$$
only when $\norm{Z_1 Z_3}>\norm{Z_1 Z_2}$.  Depending on $d_1$ and $\frac{\norm{Z_1 Z_3}}{\norm{Z_1 Z_2}}$, we have the minimum $X$ as follows:
$$
 \begin{array}{cc}
\mbox{If }\norm{Z_1 Z_3}\le\norm{Z_1 Z_2} \mbox{ or } d_1 < P,  &\mbox{then }  X=\{ S_{12+}\} \\
\mbox{If } \norm{Z_1 Z_3}>\norm{Z_1 Z_2} \mbox{ and } d_1 = P, &\mbox{then } X=\{ S_{12+}, S_{23\pm},  S_{31\pm} \}\\
\mbox{If } \norm{Z_1 Z_3}>\norm{Z_1 Z_2} \mbox{ and } d_1>P,  &\mbox{then }  X=\{ S_{23\pm}, S_{31\pm} \}
  \end{array}
 $$
\end{theorem}

The number of solutions when the measurement triangle and three measurements satisfy the assumption of Theorem \ref{th:main3}, is displayed in Fig. \ref{fig:2D5} depending on 
$\frac {\norm{Z_1 Z_2}}{\norm{Z_1 Z_3}}$  and $d_1$ (with respect to $P$) and the object values at the intersections of measurement circles are displayed in Table \ref{tab:2D5}.
 \begin{figure}
\begin{minipage}[t]{8cm}
\centerline{\epsfig{file=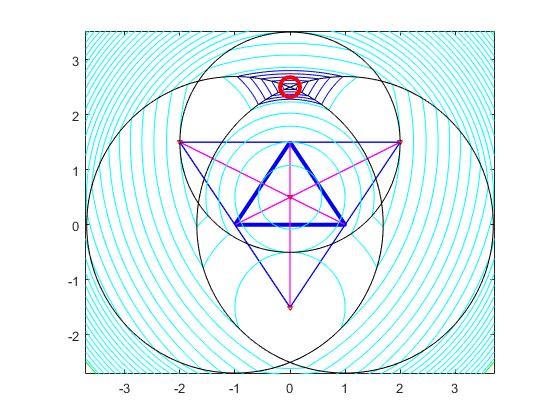, height=6cm,width=8cm,clip=1cm}}
\end{minipage}
\begin{minipage}[t]{8cm}
\centerline{\epsfig{file=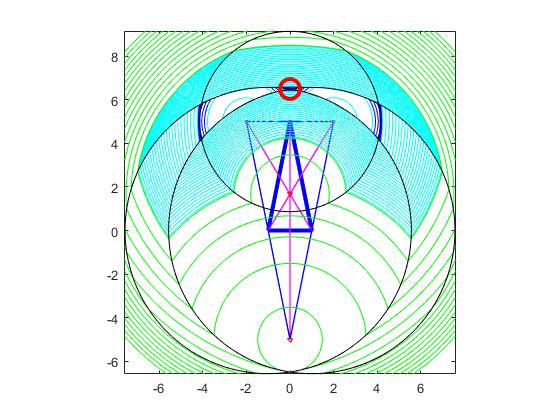, height=6cm,width=8cm,clip=1cm}}
\end{minipage}
\begin{center}
(a)\qquad\qquad\qquad\qquad\qquad\qquad\qquad\qquad\qquad\qquad(b)
\end{center}
\begin{minipage}[t]{8cm}
\centerline{\epsfig{file=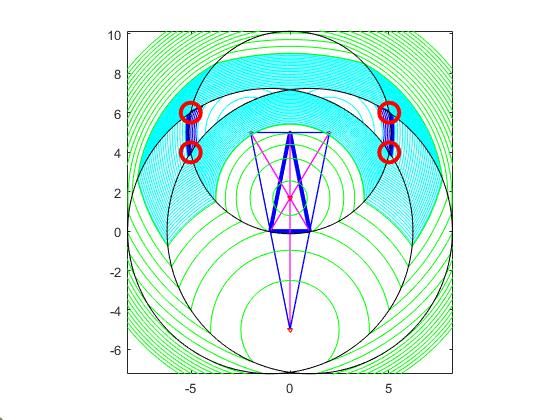, height=6cm,width=8cm,clip=1cm}}
\end{minipage}
\begin{minipage}[t]{8cm}
\centerline{\epsfig{file=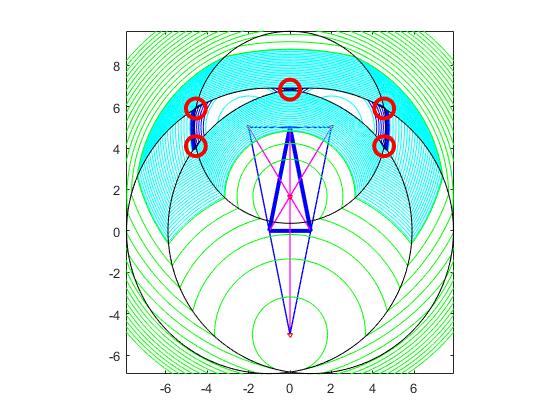, height=6cm,width=8cm,clip=1cm}}
\end{minipage}
\begin{center}
(c)\qquad\qquad\qquad\qquad\qquad\qquad\qquad\qquad\qquad\qquad(d)
\end{center}
\caption{The figure for Theorem \ref{th:main3} when $\norm{Z_1 Z_2}=2$. The solutions are (a) $X=\{S_{12+}\}$ when $\norm{Z_1 Z_3}<\norm{Z_1 Z_2}$,  (b) $X=\{S_{12+}\}$ when
$\norm{Z_1 Z_3}<\norm{Z_1 Z_2}$ and $d_1<P$,  (c) $X=\{ S_{23\pm}, S_{31\pm}\}$ when $\norm{Z_1 Z_3}<\norm{Z_1 Z_2}$ and $d_1>P$, and  (d) $X=\{ S_{12+}, S_{23\pm}, S_{31\pm}\}$ when $\norm{Z_1 Z_3}<\norm{Z_1 Z_2}$ and $d_1=P$ .} 
\label{fig:2D5}
\end{figure}
\begin{table}
\begin{centering}
\begin{tabular}{|c||c|c|c|c|c|c|}
\hline
                                  &$O(S_{12+})$           &$O(S_{23+})$            &$O(S_{31+})$           &$O(S_{12-})$   &$O(S_{23-})$             &$O(S_{31-})$\\
\hline \hline
Fig. \ref{fig:2D5} (a)   &{\color{red}3.0000}   &6.6564                         &6.6564                        &12.0000             &6.6564                        &6.6564 \\
\hline
Fig. \ref{fig:2D5} (b)   &{\color{red}14.8862} &16.2207                       &16.2207                      &114.8862           &16.2207                      &16.2207 \\
\hline
Fig. \ref{fig:2D5} (c)  &21.6750                      &{\color{red}20.1430}   &{\color{red}20.1430}  &121.6750           &{\color{red}20.1430}   &{\color{red}20.1430}\\
\hline
Fig. \ref{fig:2D5} (d)  &{\color{red}18.1818}  &{\color{red}18.1818}   &{\color{red}18.1818}   &118.1818          &{\color{red}18.1818}   &{\color{red}18.1818}\\
\hline
\end{tabular}
\caption{The comparison of objective values at $S_{12\pm}, S_{23\pm}, $ and $S_{31\pm}$ for Fig. \ref{fig:2D5}. }
\label{tab:2D5}
\end{centering}
\end{table}

From the Theorem \ref{th:main3}, we can derive the equivalent condition for $|X|=5$ as follows:
\begin{corollary}\label{co:2D5}
$ |X|=5$ if and only if 
$$\norm{Z_1 Z_3} = \norm{Z_2 Z_3},\;  \norm{Z_1 Z_3}>\norm{Z_1 Z_2}, \; d_1^2 = \frac{\norm{Z_1 Z_2}^2}{4} + \left( \norm{Z_1 Z_3}^2 - \frac{\norm{Z_1 Z_2}^2}4\right)  \left( \frac { \norm{Z_1 Z_3}^2 + \norm{Z_1 Z_2}^2} { \norm{Z_1 Z_3}^2 - \norm{Z_1 Z_2}^2}  \right)^2,$$
$$ d_3 =\sqrt{ d_1^2 - \norm{Z_1 Z_3}^2}, \; d_2 = d_1.$$
by reordering if necessary.
\end{corollary}


In Section 3, we will prove Theorem 1, 2, and 3. The lemmas for the proof of the theorems are stated and proved in Section 2. 
The number and exact locations of the solutions are presented in Section 4 and 5, when the measurment triangle is equilateral and isosceles, respectively. 
Numerical results with Matlab `contour' function are displayed with our suggested solutions. 
Tables \ref{tab:2D5}, \ref{tab:lem11}, and \ref{tab:th15} for the objective values of $S_{12\pm}, S_{23\pm},$ and $S_{31\pm}$,
show that the suggested number of solutions are correct.
 
\section{Lemmas}
Denote $S_{ij0} := \frac{S_{ij-} + S_{ij+}}2$ if $S_i$ and $S_j$ meet for $i,j=1,2,3, i\neq j$. We will notate `iff'  as an abbreviation of `if and only if'.

\begin{lem}\label{le:2DR4}
Assume that $S_1 \cap S_2 \neq \phi$. If 
\begin{equation*}
d_3^0 := \sqrt{ \frac 1 2 \norm{Z_3-S_{12+}}^2 + \frac 1 2 \norm{Z_3-S_{12-}}^2},   
\end{equation*}
then we have
\begin{equation}\label{eq:2DR4}
\left\{
\begin{array}{ccc}
O(S_{12+})<O(S_{12-})   &\mbox{ iff }& d_3  < d_{3,0} ,\\
O(S_{12+})=O(S_{12-})   &\mbox{ iff }& d_3  = d_{3,0} ,\\
O(S_{12+})>O(S_{12-})   &\mbox{ iff }& d_3  > d_{3,0} .
\end{array}\right.
\end{equation}
\end{lem}

\begin{proof}
Since $O(S_{12+}) =\left| d_3^2 - \norm{Z_3-S_{12+}}^2 \right|, \; O(S_{12-}) = \left| d_3^2 - \norm{Z_3-S_{12-}}^2 \right|, $  and 
$\norm{Z_3-S_{12+}}\le \norm{Z_3-S_{12-}}$, we have (\ref{eq:2DR4}).
\end{proof}

Let us define $d_{1,0}$ and $d_{2,0}$ likewise as in Lemma \ref{le:2DR4}.  
\begin{lem}\label{le:d10d30}
Suppose that $\norm{Z_1 Z_3}=\norm{Z_2 Z_3}$, $d_1 = d_2$, and $R_{i_1,i_2,i_3}$ is nonempty and connected for all $i_1,i_2, i _3 \in \{0,1\}$.
Then, $O(S_{23+})=O(S_{31+}), \; O(S_{23+})=O(S_{23-})$, and we have  
$$ d_3^0 = \sqrt{d_1^2 + \norm{Z_1 Z_3}^2 - \frac {\norm{Z_1 Z_2}^2}2 } $$
and 
$$ d_1^0 = d_2^0 = \sqrt{ d_1^2 + (\norm{Z_1 Z_3}^2 + d_3^2 - d_1^2) \frac{\norm{Z_1 Z_2}^2}{2\norm{Z_1 Z_3}^2} }. $$
Further, we have 
$$\begin{array}{ccc}
O(S_{23+}) < O(S_{23-})  &{\mbox iff} &  \norm{Z_1 Z_3}^2 + d_3^2 > d_1^2 ,\\
O(S_{23+})=O(S_{23-})  &{\mbox iff} &  \norm{Z_1 Z_3}^2 + d_3^2 = d_1^2 ,\\
O(S_{23+}) >O(S_{23-})  &{\mbox iff} &  \norm{Z_1 Z_3}^2 + d_3^2 < d_1^2 .
\end{array}
$$ 
\end{lem}
\begin{proof}
By the assumption, we have $S_{120} = \frac {Z_1 + Z_2}2$ and $\norm{Z_3 - S_{12\pm}} = \norm{Z_3 - S_{120}} \mp \norm{S_{120}-S_{12+}}$.
Thus, we have 
$$ d_3^0 = \sqrt{\norm{Z_3 - S_{120}}^2 + \norm{S_{12+} - S_{120}}^2}  = \sqrt{d_1^2 + \norm{Z_1 Z_3}^2 - \frac {\norm{Z_1 Z_2}^2}2 }. $$

Let us compute $d_1^0$. Since the line $\overleftrightarrow{Z_2 Z_3}$ orthogonally bisects $\overline{S_{12+} S_{12-}}$ at $S_{230}$, we have $\angle S_{23+} Z_3 Z_2 =  \angle S_{23-} Z_3 Z_2$.
\begin{eqnarray*}
  (d_1^0)^2 &=& \frac 1 2 \norm{ Z_1 -S_{31+}}^2 + \frac 1 2 \norm{ Z_1 - S_{23-}}^2 \\
           &=&\frac 1 2 d_3^2 + \frac 1 2 \norm{Z_1 Z_3}^2 - d_3 \norm{Z_1 Z_3} \cos(\angle Z_1 Z_3 Z_2  + \angle Z_2 Z_3 S_{23-}) \\ 
             &\;\;&     \frac 1 2 d_3^2 + \frac 1 2 \norm{Z_1 Z_3}^2 - d_3 \norm{Z_1 Z_3} \cos(\angle Z_1 Z_3 Z_2  - \angle Z_2 Z_3 S_{23+}) \\
           &=& d_3^2  + \norm{Z_3 Z_2}^2 - 2 d_3\norm{Z_1 Z_3} \cos(\angle Z_1 Z_3 Z_2) \cos(\angle  Z_2 Z_3 S_{23+}) \\
           &=& d_3^2  + \norm{Z_3 Z_2}^2  - 2 d_3\norm{Z_1 Z_3}\left( 1-  \frac { \norm{Z_1 Z_2}^2} {2 \norm{Z_3 Z_1}^2} \right) \frac {d_3^2 + \norm{Z_3 Z_1}^2 - d_1^2}{2d_3 \norm{Z_3 Z_1}} \\
           &=& d_1^2 +  \frac { \norm{Z_1 Z_2}^2} {2 \norm{Z_3 Z_1}^2} (d_3^2 + \norm{Z_3 Z_1}^2 - d_1^2).
 \end{eqnarray*}
Hence, $d_1<d_1^0$ iff $d_1^2 <d_3^2 + \norm{Z_3 Z_1}^2$ and we proved the lemma.
\end{proof}

Further, let us define $Y_0 := \frac 1 3 \sum_i Z_i$ and  $Y_j = 3Y_0 - 2Z_j$ for $j=1,2,3$.

\begin{lem}\label{le:2D-}
Suppose that $S_{23-}, S_{31-} \in R_0$. Then, $O(S_{23-})=O(S_{31-})$ iff  
$$ \norm{S_{23-} - Y_0} = \norm{S_{31-} - Y_0}, \mbox{ and } \angle Z_3 Y_0 S_{23-}  = \angle Z_3 Y_0 S_{31-}. $$ 
Suppose that $S_{12-}, S_{23-}, S_{31-}  \in R_0$. Then, $O(S_{23-})=O(S_{31-})=O(S_{12-})$ iff  
$$ \norm{S_{23-} - Y_0} = \norm{S_{31-} - Y_0} = \norm{S_{12-} - Y_0}, $$
and
\begin{subeqnarray}
\slabel{eq:-Angle1} \angle Z_3 Y_0 S_{23-}  &=& \angle Z_3 Y_0 S_{31-} = \pi - \angle Z_1 Y_0 Z_2, \\
\slabel{eq:-Angle2} \angle Z_1 Y_0 S_{31-}  &=& \angle Z_1 Y_0 S_{12-} = \pi - \angle Z_2 Y_0 Z_3, \\
\slabel{eq:-Angle3} \angle Z_2 Y_0 S_{12-}  &=& \angle Z_2 Y_0 S_{23-} = \pi - \angle Z_3 Y_0 Z_1. 
\end{subeqnarray}
\end{lem}

\begin{proof}
Suppose that $S_{23-}, S_{31-} \in R_0$. From the condition that $O(S_{23-})=O(S_{31-})$, we have $\norm{Z_3 - S_{31-}} = \norm{Z_3 - S_{23-}}$.
If $W\in R_0$, we have 
\begin{subeqnarray*}
O(W) &=& \norm{W-Z_1}^2 - d_1^2  +  \norm{W-Z_2}^2 - d_2^2  +  \norm{W-Z_3}^2 - d_3^2 \\
         &=& 3\norm{W-Y_0}^2  - 3\norm{Y_0}^2  -d_1^2 - d_2^2 - d_3^2 + \norm{Z_1}^2 + \norm{Z_2}^2 + \norm{Z_3}^2.
\end{subeqnarray*}         
Since $S_{23-}, S_{31-} \in R_0$,  we have $\norm{S_{23-} - Y_0} = \norm{S_{31-} - Y_0}$. 
From $\norm{Z_3 - S_{31-}} = \norm{Z_3 - S_{23-}}$ and $\norm{S_{23-} - Y_0} = \norm{S_{31-} - Y_0}$, we have  
$$  \angle Z_3 Y_0 S_{23-}  = \angle Z_3 Y_0 S_{31-} ,$$
as in Fig. \ref{fig:2D+-} (a). The converse can be proved easily.

Suppose that $S_{12-}, S_{23-}, S_{31-}  \in R_0$.  
From the condition that $O(S_{23-})=O(S_{31-})=O(S_{12-})$, we have $\norm{Z_3 - S_{23-}} = \norm{Z_3 - S_{31-}},  \norm{Z_1 - S_{31-}} = \norm{Z_1 - S_{12-}},$ and $\norm{Z_2 - S_{12-}} = \norm{Z_2 - S_{23-}}.$
Since $S_{12-}, S_{23-}, S_{31-}  \in R_0$, with the similar argument as above, we have
$$ \norm{S_{23-} - Y_0} = \norm{S_{31-} - Y_0} = \norm{S_{12-} - Y_0}. $$
And further with similar argument, we have   
$$  \angle Z_3 Y_0 S_{23-}  = \angle Z_3 Y_0 S_{31-} , \quad \angle Z_1 Y_0 S_{31-}  = \angle Z_1 Y_0 S_{12-} , \quad  \angle Z_2 Y_0 S_{12-}  = \angle Z_2 Y_0 S_{23-} . $$
Since $ 2( \angle Z_3 Y_0 S_{23-} + \angle Z_1 Y_0 S_{31-} + \angle Z_2 Y_0 S_{12-}) = 2\pi$, we have 
$$ \angle Z_3 Y_0 S_{23-} = \pi - (\angle Z_1 Y_0 S_{31-} + \angle Z_2 Y_0 S_{12-}) = \pi - (\angle Z_1 Y_0 S_{12-} + \angle Z_2 Y_0 S_{12-}) = \pi - \angle Z_1 Y_0 Z_2.$$
The two equations \eqref{eq:-Angle2} and \eqref{eq:-Angle3} can be derived similarly. The converse can be proved without difficulty.
\end{proof}

 \begin{figure}
\begin{minipage}[t]{8cm}
\centerline{\epsfig{file=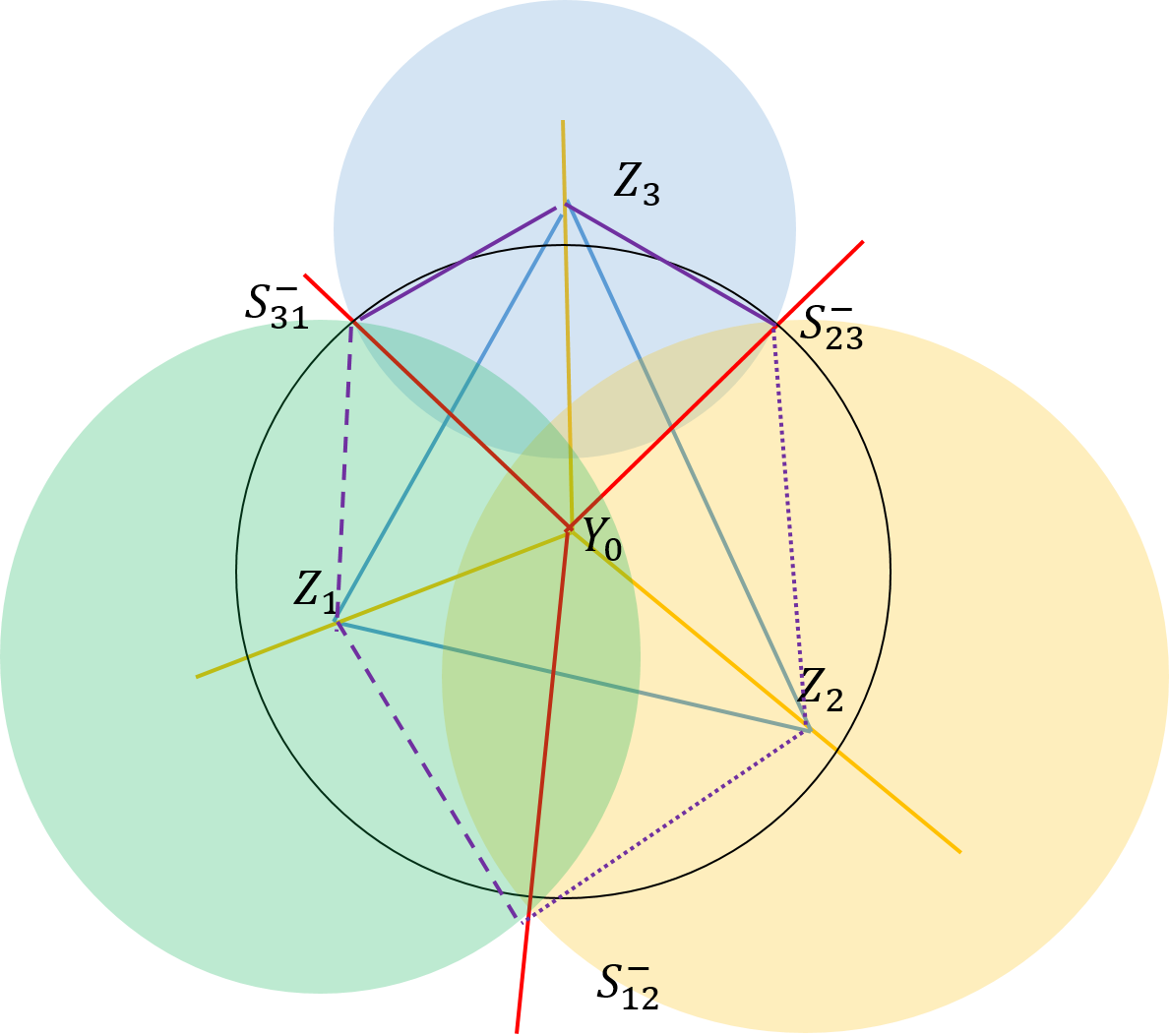, height=6cm,width=7cm,clip=1cm}}
\end{minipage}
\begin{minipage}[t]{8cm}
\centerline{\epsfig{file=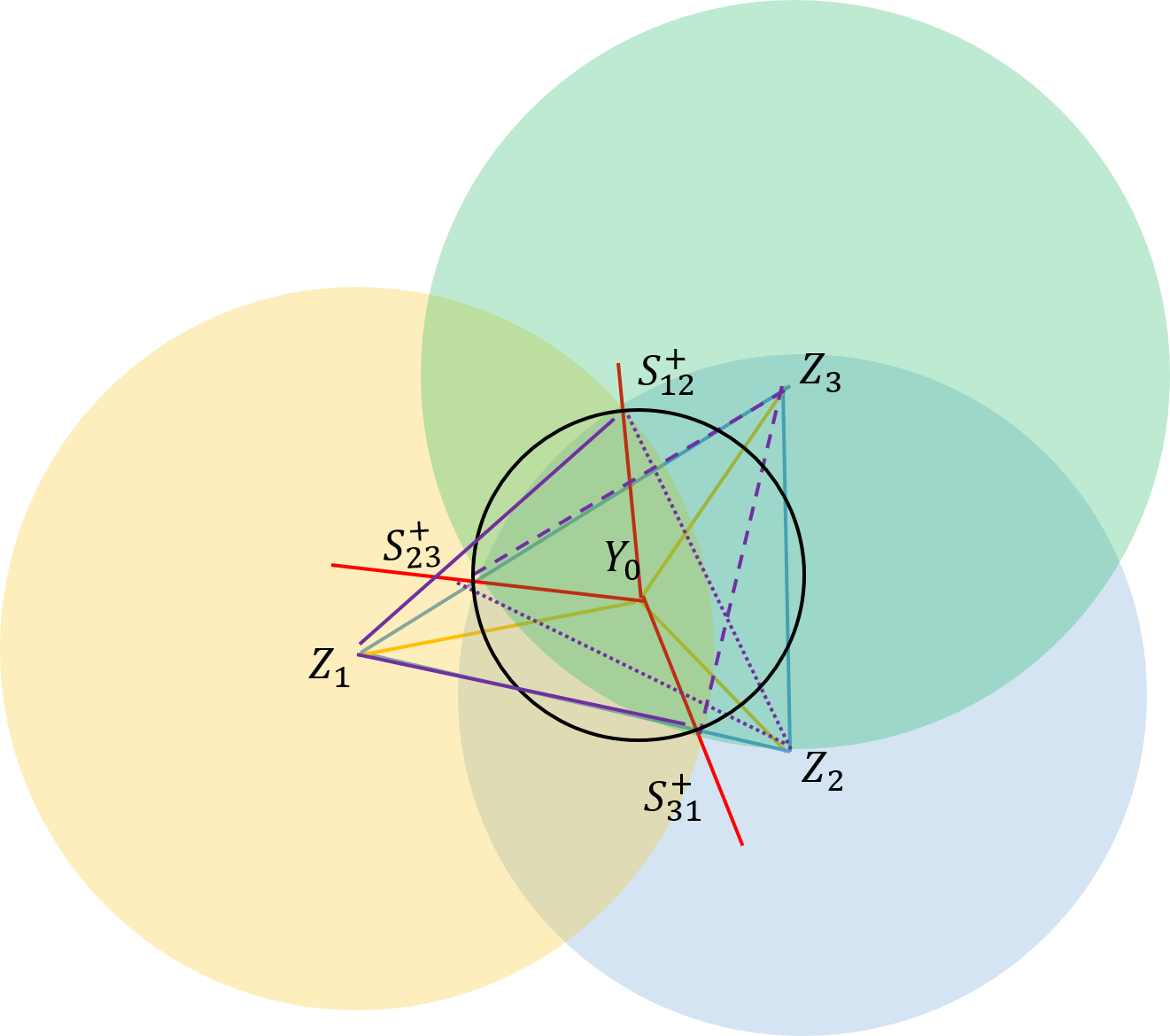, height=6cm,width=7cm,clip=1cm}}
\end{minipage}
\begin{center}
(a)\qquad\qquad\qquad\qquad\qquad\qquad\qquad\qquad\qquad\qquad(b)
\end{center}
\caption{The figure for (a) Lemma \ref{le:2D-} and (b) Lemma \ref{le:2D+}.}
\label{fig:2D+-}
\end{figure}

\begin{lem}\label{le:2D+}
If $S_{23+}, S_{31+} \in R_3$ and  $O(S_{23+})=O(S_{31+})$, then 
$$ \norm{S_{23+} - Y_0} = \norm{S_{31+} - Y_0}, $$
and
$$  \angle Z_3 Y_0 S_{23+}  = \angle Z_3 Y_0 S_{31+}. $$ 
If $S_{12+}, S_{23+}, S_{31+} \in R_3$ and $O(S_{23+})=O(S_{31+})=O(S_{12+})$, then 
$$ \norm{S_{23+} - Y_0} = \norm{S_{31+} - Y_0} = \norm{S_{12+} - Y_0},$$
and
\begin{subeqnarray}\label{eq:2D+}
\slabel{eq:2D+1} \angle Z_3 Y_0 S_{23+}  &=& \angle Z_3 Y_0 S_{31+} =  \angle Z_1 Y_0 Z_2, \\
\slabel{eq:2D+2}  \angle Z_1 Y_0 S_{31+}  &=& \angle Z_1 Y_0 S_{12+} =  \angle Z_2 Y_0 Z_3, \\
\slabel{eq:2D+3}  \angle Z_2 Y_0 S_{12+}  &=& \angle Z_2 Y_0 S_{23+} =  \angle Z_2 Y_0 Z_3. 
\end{subeqnarray}
\end{lem}

\begin{proof}
We can prove the Lemma \ref{le:2D+} with the same arguement as in Lemma \ref{le:2D-}.
The only difference is the last equality in (\ref{eq:2D+}).  We have the following equations when $O(S_{23+})=O(S_{31+})=O(S_{12+})$:
$$
 \angle Z_3 Y_0 S_{23+}  = \angle Z_3 Y_0 S_{31+} , \; \angle Z_1 Y_0 S_{31+}  = \angle Z_1 Y_0 S_{12+} ,  \;  \angle Z_2 Y_0 S_{12+}  = \angle Z_2 Y_0 S_{23+} .
$$
And we can derive the followings (see Fig. \ref{fig:2D+-} (b)) : 
\begin{eqnarray*}
 \angle Z_1 Y_0 S_{12+}  + \angle Z_2 Y_0 S_{12+}  &=& \angle Z_1 Y_0 Z_3 + \angle Z_2 Y_0 Z_3, \\
 \angle Z_2 Y_0 S_{23+} +  \angle Z_3 Y_0 S_{23+} &=& \angle Z_3 Y_0 Z_1  + \angle Z_2 Y_0 Z_1, \\
  \angle Z_3 Y_0 S_{31+}  +  \angle Z_1 Y_0 S_{31+} &=& \angle Z_3 Y_0 Z_2 + \angle Z_1 Y_0 Z_2. 
\end{eqnarray*}
Hence we have $ \angle Z_1 Y_0 S_{12+} +  \angle Z_2 Y_0 S_{23+} +  \angle Z_3 Y_0 S_{31+} =  \angle Z_1 Y_0 Z_2 + \angle Z_2 Y_0 Z_3 + \angle Z_3 Y_0 Z_1  = 2\pi$.   
From these equalities, we can derive the following:
$$ \angle Z_3 Y_0 S_{23+} = 2\pi - \left( \angle Z_1 Y_0 S_{31+}  + \angle Z_2 Y_0 S_{12+} \right) = \angle Z_1 Y_0 S_{12+} + \angle Z_2 Y_0 S_{31+} =  \angle Z_1 Y_0 S_{31+} + \angle Z_2 Y_0 S_{31+} = \angle Z_1 Y_0 Z_2. $$
The other two equalites \eqref{eq:2D+1} and \eqref{eq:2D+1} can be derived similarly. 
\end{proof}

Let us notate $L=\frac{Z_1+Z_2}2, M=\frac{Z_2 +Z_3}2, $ and $N=\frac{Z_3+Z_1}2$.
\begin{lem}\label{le:2D22}
The condition that $R_{i1,i2,i3}$ is nonempty and connected for all $i_1,i_2 \in \{0,1\}$ and  
$$ O(S_{23+})=O(S_{23-})=O(S_{31+})=(S_{31-})$$ 
is equivalent to the following condition
$$\norm{Z_1 Z_3} = \norm{Z_2 Z_3}, d_1 = d_2, \mbox{ and } d_1^2 = d_3^2 + \norm{Z_1 Z_3}^2.$$
\end{lem}
\begin{proof}
Assume that $R_{i1,i2,i3}$ is nonempty and connected for all $i_1,i_2 \in \{0,1\}$ and 
$$ O(S_{23+})=O(S_{23-})=O(S_{31+})=(S_{31-}).$$
Possible locations of $S_{31+}$ and $S_{23+}$ are in Figure \ref{fig:2D22} (a) and (b).
By using Lemma \ref{le:2D-}  and  Lemma \ref{le:2D+},  we have 
$$\angle Y_0 Z_3 S_{31+} = \angle Y_0 Z_3 S_{23+}, \mbox{ and } \angle S_{31-} Z_3 Y_0 = \angle S_{23-} Z_3 Y_0. $$
From these facts,  we have 
\begin{equation} \label{eq:twosides}
 \angle Z_1 Z_3 Y_0 = \angle Z_2 Z_3 Y_0.
 \end{equation}
 
From (\ref{eq:twosides}) and $\norm{Z_1 L}=\norm{Z_2 L}$, we have 
$\norm{Z_1 Z_3} = \norm{Z_2 Z_3}$ by the angle bisector theorem.
Since 
$$\angle Z_1 S_{31-} Z_3 = \frac 1 2  \angle Z_1 Z_3 Y_0  =   \frac 1 2 \angle Z_2 Z_3 Y_0 = \angle Z_2 S_{31-} Z_3,$$
$\triangle Z_1 S_{31-} Z_3$ and $\triangle Z_2 S_{23-} Z_3$ are congruent and we can derive that $d_1=d_2$.
By Lemma \ref{le:d10d30}, we have
$d_1^2 = d_3^2 + \norm{Z_3 Z_1}^2$.

From the condition $\norm{Z_1 Z_3} = \norm{Z_2 Z_3}, d_1 = d_2, \mbox{ and } d_1^2 = d_3^2 + \norm{Z_1 Z_3}^2$, we can prove that
$$ d_1 > \frac{\norm{Z_1 Z_2}}2 \mbox{ and } \sqrt{d_1^2 -\frac{\norm{Z_1 Z_2}^2}4}-\sqrt{\norm{Z_1 Z_3}^2  -\frac{\norm{Z_1 Z_2}^2}4} < d_3 < \sqrt{d_1^2 +\frac{\norm{Z_1 Z_2}^2}4}-\sqrt{\norm{Z_1 Z_3}^2  -\frac{\norm{Z_1 Z_2}^2}4}, $$
which proves that $R_{i1,i2,i3}$ is nonempty and connected for all $i_1,i_2 \in \{0,1\}$. We used the fact that if $a>b>0$ then $a-b < \sqrt{a^2-b^2} < a+b$. 
We can prove $ O(S_{23+})=O(S_{23-})=O(S_{31+})=(S_{31-})$ by Lemma \ref{le:d10d30}. 
\end{proof}

 \begin{figure}
\begin{minipage}[t]{8cm}
\centerline{\epsfig{file=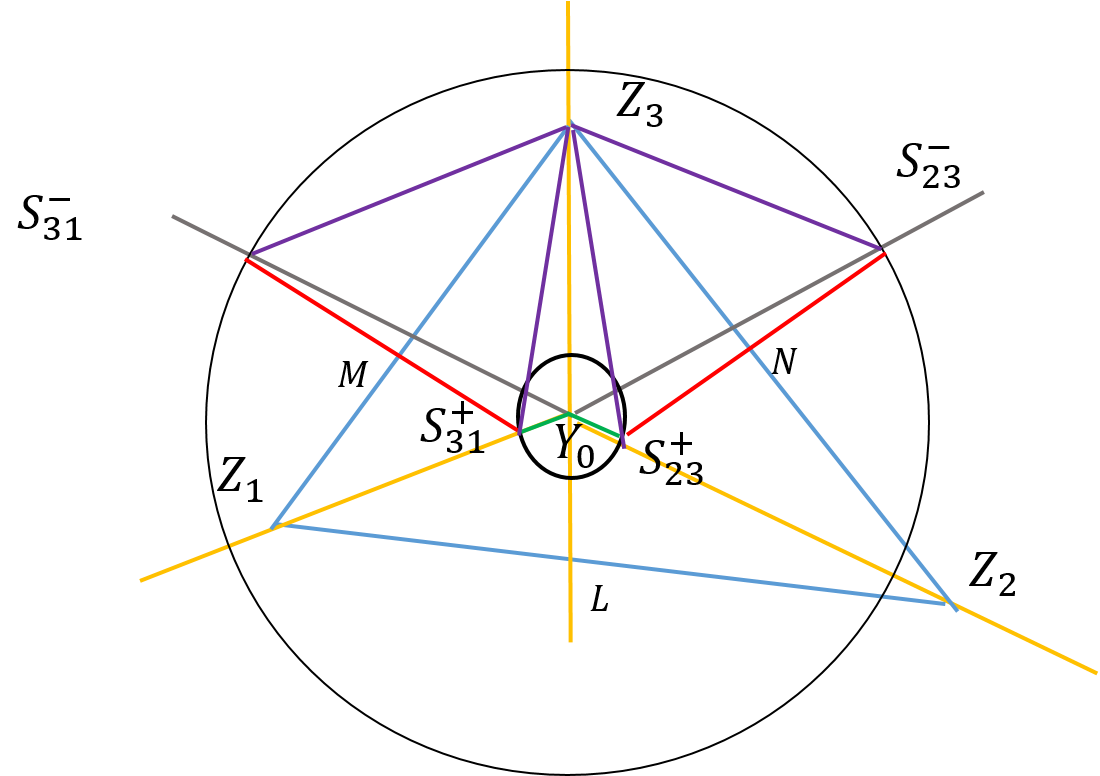, height=6cm,width=7cm,clip=1cm}}
\end{minipage}
\begin{minipage}[t]{8cm}
\centerline{\epsfig{file=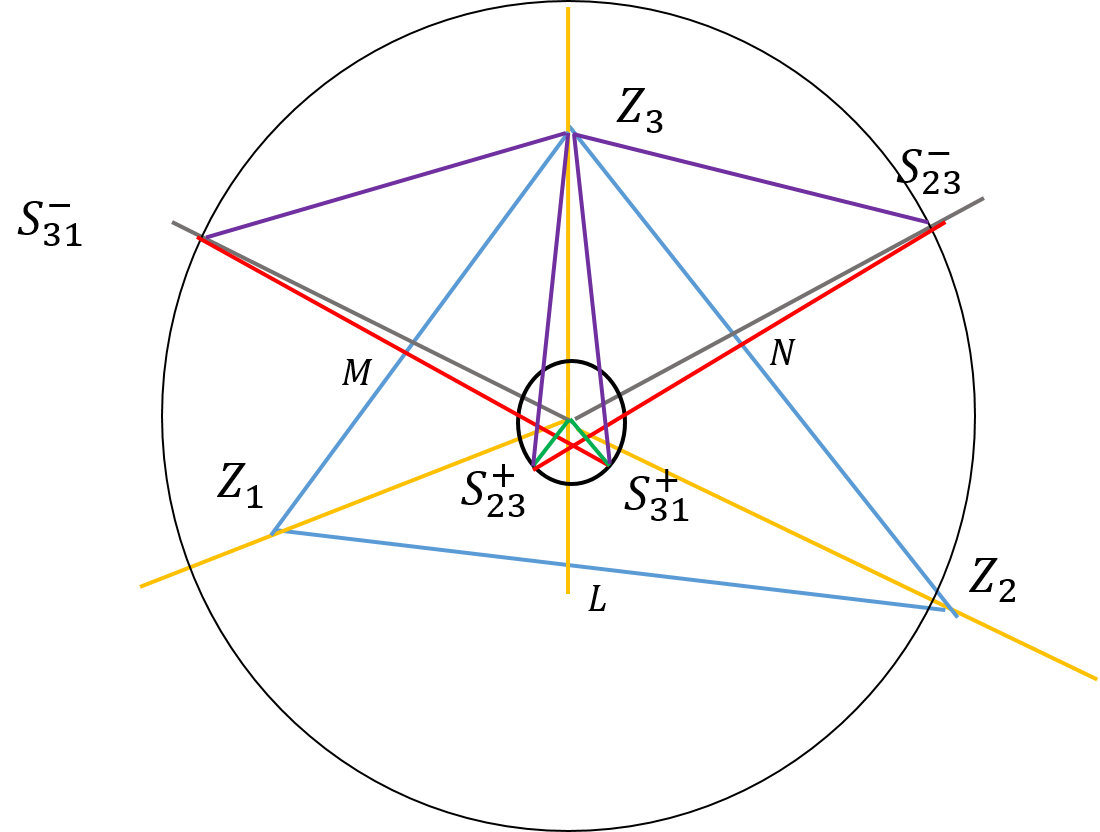, height=6cm,width=7cm,clip=1cm}}
\end{minipage}
\begin{center}
(a)\qquad\qquad\qquad\qquad\qquad\qquad\qquad\qquad\qquad\qquad(b)
\end{center}
\caption{The figure for Lemma \ref{le:2D22}.} 
\label{fig:2D22}
\end{figure}

Let us notate $r:= \norm{Z_1 Z_2}$ and $s:=\norm{ Z_3 - L}$. Then, in an equilateral triangle, $s=\frac{\sqrt 3}2 r$.
\begin{lem}\label{le:equid123}
Suppose that the measurement triangle $\triangle Z_1 Z_2 Z_3$ is equilateral and $d_1 = d_2 = d_3$. Then, the sources $X$ are defined as follows:
$$
X = \left\{
\begin{array}{ccc}
Y_0  &   \mbox{ if } & d_1 \le \frac r {\sqrt 3},               \\
\{ S_{12+}, S_{23+}, S_{31+} \} &  \mbox{ if } & d_1 > \frac r{\sqrt 3}.
\end{array}
\right.
$$
\end{lem}
\begin{proof}
If $d_1 \le  \frac r{\sqrt 3}$, then $Y_0 \in R_0$ and  $X= Y_0$.

If $d_1 > \frac r{\sqrt 3}$, then $Y_0 \in R_3$ and  by symmetry, we have $O(S_{31+}) = O(S_{23+}) = O(S_{12+})$ and $O(S_{31-}) = O(S_{12-}) = O(S_{23-})$. 
Since the measurement triangle is equilateral, we have $d_1^0 = d_2^0 = d_3^0 = \sqrt{d_1^2 + \frac {r^2} 2} > d_1=d_2=d_3$. Hence, 
$$ O(S_{31+}) = O(S_{23+}) = O(S_{12+}) < O(S_{31-}) = O(S_{12-}) = O(S_{23-}),$$
by Lemma \ref{le:2DR4}, and we have $X = \{S_{12+}, S_{23+}, S_{31+} \}$.
\end{proof}

\section{The proofs of theorems}
\subsection{Theorem \ref{th:main1}}
To prove the theorem, it is required to correct the proof of Theorem 6(c) when $R_{100},R_{010},$ and $R_{001}$ are connected and nonempty. 
Let us denote that $X_{i1, i2, i3}$ and $X_j$ are the minima in the domain $R_{i1,i2,i3}, \; i1,i2,i3 \in \{0,1\}$ and $R_j, \; j=0,1,2,3$, respectively.

Assume that $R_{100},R_{010},$ and $R_{001}$ are all nonempty and connected. Then, we have $Y_0\notin R_0, Y_1\notin R_{100}, Y_2\notin R_{010},$ and $Y_3\notin R_{001}$.
Since $X_{001}$ is the closest point to $Y_1$, $X_{001}\subset (S_2\cup S_3)\cap R_{100}\subset R_2$. In a similar manner, we have $X_{010}, X_{001}\subset R_2$. Hence, we have $X=X_2$ {\color{red}$\cup X_3$}.
Since $X_{011}$ is the farthest point in $R_{011}$ from $Y_1$, $X_{101}$ is the farthest point in $R_{101}$ from $Y_2$, $X_{110}$ is the farthest point in $R_{110}$ from $Y_3$, and $X_{111}$ is the farthest point in $R_{111}$ from $Y_0$,
we can derive that $X = argmin_{A\in S_{123+} {\color{red} \cup S_{123-}}} \norm{A-Y_0}$. 

The converse can be proved using Theorem 5, Theorem 6 (execpet the case $R_{100},R_{010},$ and $R_{001}$ are all connected and nonempty), Lemma 12, Lemma 13, and Corollary 14  in \cite{Kwon}.

\subsection{Theorem \ref{th:main2}}
It is enough to show that $|X|\neq 6$ by Theorem \ref{th:main1}.
The only case that $|X|=6$ is that 
$$  O(S_{12+})=O(S_{12-})=O(S_{23+})=O(S_{23-})=O(S_{31+})=O(S_{31-}).$$
Then, by Lemma \ref{le:2D-} and \ref{le:2D+}, we have, 
$$ \norm{S_{23-} - Y_0} = \norm{S_{31-} - Y_0} = \norm{S_{12-} - Y_0},\quad \norm{S_{23+} - Y_0} = \norm{S_{31+} - Y_0} = \norm{S_{12+} - Y_0}.$$
and 
\begin{eqnarray*}
 \angle Z_1 Y_0 S_{12-}  + \angle Z_1 Y_0 S_{12+}  = \pi , &\;& \angle Z_1 Y_0 S_{31-}  + \angle Z_1 Y_0 S_{31+}  = \pi , \\
\angle Z_2 Y_0 S_{23-}  + \angle Z_2 Y_0 S_{23+}  = \pi , &\;& \angle Z_2 Y_0 S_{12-}  + \angle Z_2 Y_0 S_{12+}  = \pi , \\
\angle Z_3 Y_0 S_{31-}  + \angle Z_3 Y_0 S_{31+}  = \pi , &\;& \angle Z_3 Y_0 S_{23-}  + \angle Z_3 Y_0 S_{23+}  = \pi ,
\end{eqnarray*}
which implies $ \overline{S_{ij+} S_{ij-}}$ and $\overline{ Z_i Z_j}$ are othogonal each other for all $(i,j)=(1,2),(2,3),$ and $(3,1)$.
Hence, the distance from $Y_0$ to each side of the triangle $\triangle Z_1 Z_2 Z_3$ is the same and  we have
\begin{equation}\label{eq:angle}
 \angle Y_0 Z_1 Z_2 = \angle Y_0 Z_1 Z_3. 
 \end{equation}
Let $M= \frac 1 2 ( Z_2 + Z_3)$. Then, the line joining $Z_1$ and $Y_0$ passes $M$. And the from (\ref{eq:angle}), we have
$$ \norm{Z_2 - Z_1} : \norm{Z_1 - Z_3} = \norm{Z_2 - M} : \norm{Z_3 - M} = 1: 1. $$

Hence, the triange $\triangle Z_1 Z_2 Z_3$ should be an equilateral triangle and $d_1=d_2=d_3$.
Then by Lemma \ref{le:equid123}, we have $X=S_{123+}$ iff $d_1 > \frac {|Z_1-Z_2|}{\sqrt 3}$ and $X=Y_0$, otherwise.
In any of the two cases, this contradicts that $|X|=6$. Hence, we proved the theorem.

\subsection{Theorem \ref{th:main3}}
By Lemma \ref{le:2D22}, the assumption implies that $R_{i1,i2,i3}$ is nonempty and connected for all $i_1, i_2 \in \{0,1\}$ and 
$$O(S_{23+})=O(S_{23-})=O(S_{31+})=O(S_{31-}).$$
Let us compute one of them: 
\begin{eqnarray*}
 O(S_{31-}) &=& \norm{S_{31-}  Z_2}^2 - d_2^2 = d_3^2 + \norm{Z_3 Z_2}^2 - 2 d_3 \norm{Z_3 Z_2} \cos\angle S_{31-} Z_3 Z_2 - d_2^2\\
                    &=&   d_3^2 + \norm{Z_3 Z_2}^2 + 2 d_3 \norm{Z_3 Z_2} \sin\angle Z_1  Z_3 Z_2 - d_2^2\\
                    &=& (d_3^2 + \norm{Z_3 Z_2}^2 - d_2^2)  + 2 d_3 \norm{Z_3 Z_2} 2\sin \frac{\angle Z_1  Z_3 Z_2}2 \cos \frac{\angle Z_1  Z_3 Z_2} 2\\
                    &=&  \frac{2r sd_3}{\sqrt{s^2  +\frac{r^2}4}  }
\end{eqnarray*}
where $\angle S_{31-} Z_3 Z_2 = \frac \pi 2 + \angle Z_1  Z_3 Z_2 $ is used. Whereas,
\begin{eqnarray*}
O(S_{12-}) &=& \norm{S_{12-}  Z_3}^2 - d_3^2 = \left( \sqrt{d_1^2 - \frac{r^2}4} + s\right)^2 - d_3^2 \\
                   &=& d_1^2 + s^2 - \frac{r^2}4 - d_3^2 + 2 s\sqrt{d_1^2 - \frac{r^2}4}\\
                   &=& 2s^2  + 2 s\sqrt{d_3^2 + s^2}.
\end{eqnarray*}  

Likewise, we obtain                 
$$ O(S_{12+}) = d_3^2 - \norm{S_{12+}  Z_3}^2 = -2s^2  + 2 s\sqrt{d_3^2 + s^2}. $$

Since $O(S_{12+})< O(S_{12-})$, it is enough to compare $O(S_{12+})$ and $O(S_{31-})$ to prove the lemma.

First, suppose that $\norm{Z_1 Z_2} \ge \norm{Z_1 Z_3}$. Then, 
$$ O(S_{12+}) = 2s \left(\sqrt{d_3^2 + s^2} - s\right) <2sd_3\le O(S_{31-}).$$

Next, suppose that $\norm{Z_1 Z_2} < \norm{Z_1 Z_3}$. If we find the relation between $d_3$ and $s$ satisfying $O(S_{12+}) = O(S_{31-})$, we have
$$ d_3 = \frac{2rs\sqrt{s^2 + \frac{r^2}4}}{s^2-\frac{3r^2}4}. $$
Using the relation $d_1^2 = d_3^2 + \frac{r^2} 4 + s^2$,  the above equation is equivalent to 
$$ d _1 = \sqrt{\frac{r^2}{4} +   s^2 \left( \frac { s^2 + \frac 5 4 r^2} { s^2 - \frac 3 4 r^2}  \right)^2}. $$
Since we have defined the righthandside as $P$, if we follow the same computation carefully replacing $=$ with $>$ and $<$, we can prove the theorem.

\section{When the measurement triangle is equilateral}
\begin{theorem}\label{le:equid12}
Suppose that the measurement triangle $\triangle Z_1 Z_2 Z_3$ is equilateral and $d_1=d_2$. 
Then, the source $X$ is determined as follows: 

\begin{tabular}{|c|c||c|}
\hline
$ d_1=d_2$                                          &$d_3$                                                                                                                                                         &$X$ \\
\hline \hline
$\left(0,\frac r{\sqrt 3}\right]$              &$\left(0, \frac r{\sqrt 3}\right]$                                                                                                                     &$Y_0$\\
\hline
$\left(0,r\right]$                                    &$\left[ \sqrt 3 r, \infty \right)$                                                                                                                       &$Y_3$\\
\hline
$\left(0, \frac r 2\right]$                       &$\left[\frac r{\sqrt 3},\sqrt 3 r \right]$                                                                                                           &$N_3$\\
\hline
$\left(\frac r 2, \frac r{\sqrt 3}\right]$  &$\left[\frac r{\sqrt 3}, \frac{\sqrt 3}2 r- \sqrt{ d_1^2 - \frac{r^2}4 }\right]$                                                   &$N_3$\\
\hline
$\left( \frac r 2, r\right]$                       &$\left[ \frac{\sqrt 3}2 r + \sqrt{ d_1^2 - \frac{r^2}4},   \sqrt 3 r \right]$                                                        &$N_3$\\    
\hline
$\left(\frac r 2, \frac r{\sqrt 3}\right)$  &$\left(\frac{\sqrt 3}2 r - \sqrt{ d_1^2 - \frac{r^2}4}, \frac{\sqrt 3}2 r + \sqrt{ d_1^2 - \frac{r^2}4 }\right)$   &$\{S_{23+}, S_{31+}\}$  \\
\hline
$\left[\frac r{\sqrt 3},\infty\right)$        &$\left(0, d_1\right)$                                                                                                                                       &$S_{12+}$\\
\hline    
                                                            &$d_1$                                                                                                                                                            &$\{ S_{12+}, S_{23+}, S_{31+} \}$\\                                                
\hline
$\left[\frac r{\sqrt 3}, r\right)$              &$\left( d_1, \frac{\sqrt 3}2 r+ \sqrt{ d_1^2 - \frac{r^2}4} \right)$                                                                  &$\{S_{23+}, S_{31+}\}$\\
\hline                    
$\left[ r, \infty\right)$                            &$\left( d_1,  \sqrt{d_1^2 + 2r^2}   \right)$                                                                                                     &$\{ S_{23+}, S_{31+} \}$\\
\hline
                                                            &$ \sqrt{d_1^2 + 2r^2}$                                                                                                                                  &$\{ S_{23+}, S_{31+}, S_{12-} \}$\\
\hline 
                                                            &$\left(  \sqrt{d_1^2 + 2r^2},  \infty\right) $                                                                                                     &$S_{12-}$\\                                  
\hline                                                                                                                                                                                                                                                                                                                                
\end{tabular}
\end{theorem}
\begin{proof}
If $d_1 = d_2 \le \frac r{\sqrt 3} $ and $d_3 \le \frac r{\sqrt 3}$, then $Y_0\in R_0$ and $X=Y_0$ by Corollary 14(a) in \cite{Kwon}.

If $ d_1  \le  r$ and $d_3 \ge \sqrt 3 r$, then $Y_3\in R_{001}$ and $X=Y_3$ by Corollary 14(b) in \cite{Kwon}.

If $d_1< \frac r 2$, then two disks $D_1$ and $D_2$ do not meet and $R_3=\phi$. Then, $X=X_{001}$ by Theorem 6(b) in \cite{Kwon}. Since $\norm{Z_3 - Y_0} = \frac r{\sqrt 3}$ and $\norm{Z_3 - Y_3}=\sqrt 3 r$, 
if $d_3  \in [\frac r{\sqrt 3}, \sqrt 3 r]$, then $Y_0$ is contained in $R_{001}$, $Y_3$ is not contained inside of $R_{001}$, and $N_3$ is the closest point on $R_{001}$ from $Y_3$. Hence, $X=N_3$. We can think about $d_1=\frac r 2$ as the limiting case.

If two disks $D_1$ and $D_2$ meet and does not meet $D_3$ having $Y_0$ inside, then $X=N_3$. The condition of $d_1$ and $d_3$ for this case is that $d_1 \in \left[\frac r 2, \frac r{\sqrt 3}\right)$ and
$d_3 \in \left(\frac r{\sqrt 3}, \frac{\sqrt 3}2 r- \sqrt{ d_1^2 - \frac{r^2}4 }\right)$. We could insert the end points of the interval without difficult consideration.

If $d_1\in \left[ \frac r 2, r \right]$ and $d_3\in \left[  \frac{\sqrt 3}2 r+ \sqrt{ d_1^2 - \frac{r^2}4}, \sqrt 3 r\right]$, then $D_3$ contains $D_1\cup D_2$ or $R_{001}$ is nonempty and disconnected. 
Then, by Theorems 6(a) and (c) in \cite{Kwon3}, we have $X=X_{001}$. Since $Y_3$ is not contained in any of the three disks, $X_{001}$ should be the nearest point of $R_{001}$ from $Y_3$. Hence, we have $X=X_{001}=N_3$.

Consider the case $d_1\in \left( \frac r 2, \frac r {\sqrt 3} \right)$ and $d_3 \in \left( \frac{\sqrt 3}2 r-\sqrt{ d_1^2 - \frac{r^2}4}, \frac{\sqrt 3}2 r+ \sqrt{ d_1^2 - \frac{r^2}4}\right)$. 
All three disks meet and $R_{i_1 i_2 i_3}$ has one connected component for every $i_1,i_2,i_3 \in \{0,1\}$. Thus, we have $X\ge S_{123}^\pm$ by Theorem \ref{th:main1}.
Since 
$$
d_3 > \frac{\sqrt 3}2 r-\sqrt{ d_1^2 - \frac{r^2}4} >  \frac{\sqrt 3}2 r  -\sqrt{ \left(\frac r {\sqrt 3}\right)^2 - \frac{r^2}4}  =  \frac r {\sqrt 3} > d_1, 
$$
and using Lemma \ref{le:d10d30}, we have $X\ge \{S_{12\pm}, S_{23+}, S_{31+}\}$.  
Since $Y_0 \in D_3\setminus R_3$, $\norm{Z_3 - S_{12+}} < \norm{Z_3 - S_{23+}}$.
By the fact that minima of $R_3$ are located on the farthest points from $Y_0$, we have $R_3\ge \{ S_{23+}, S_{31+} \}$ and $X\ge \{ S_{23+}, S_{31+}, S_{12-} \}$.   Consider $R_{110}$. The minima of $R_{110}$ are located at the farthest points from $Y_3$ and $Y_3 \notin R_{110}$, then we have $S_{12-} \ge R_{110} \ge \{S_{23+}, S_{31+}\}$ and thus finally we have $X=\{S_{23+}, S_{31+}\}$. 

If $\frac r {\sqrt 3}< d_1 < r $ and $d_3 < \frac{\sqrt 3}2 r- \sqrt{ d_1^2 - \frac{r^2}4}$, then $R_0\cap \triangle Z_1 Z_2 Z_3 \neq \phi$ and $X=S_{12+}$ by Theorem 5(b) and Lemma 12(c) in \cite{Kwon3}. 
If $ d_1 >   r $ and $d_3 <  \sqrt{ d_1^2 - \frac{r^2}4}-\frac{\sqrt 3}2 r$,  then $D_3 \subset D_1\cap D_2$ and $X=S_{12+}$, by Theorem 6(a) and Lemma 13(c) in \cite{Kwon3}. Thus, we have proved that
$X=S_{12+}$ if  $d_1>\frac r {\sqrt 3}$ and $d_3 < \frac{\sqrt 3}2 r- \sqrt{ d_1^2 - \frac{r^2}4}$.

If $ \frac r {\sqrt 3}<d_1<r$ and $\left|\frac{\sqrt 3}2 r- \sqrt{ d_1^2 - \frac{r^2}4}\right| < d_3 < \frac{\sqrt 3}2 r+ \sqrt{ d_1^2 - \frac{r^2}4}$, $R_{i_1 i_2 i_3}$ is nonempty and connected for all $i_1,i_2,i_3\in \{0,1\}$. Thus, by Theorem \ref{th:main1}, we have $X\subset S_{123+}$.

By Lemma \ref{le:d10d30} and considering the fact that  $X_3$ is the farthest point from $Y_0$, we have
\begin{subeqnarray}\label{eq:equid10d30}
\slabel{eq:equid10d30_2}
O(S_{23+}) < O(S_{23-})                                            &\mbox{ iff }& d_1\le r \mbox{ or } \left( d_1>r \mbox{ and } d_3>\sqrt{d_1^2-r^2}\right) ,\\
\slabel{eq:equid10d30_1}
O(S_{12+}) < O(S_{12-})                                             &\mbox{ iff }& d_3<\sqrt{d_1^2 + \frac{r^2}2} ,\\
\slabel{eq:equid10d30_3}
O(S_{12+}) < O(S_{23+})                                           &\mbox{ iff }& d_3<d_1,
\end{subeqnarray}
under the condition $d_1 > \frac r {\sqrt 3}$ and $\left|\frac{\sqrt 3}2 r- \sqrt{ d_1^2 - \frac{r^2}4}\right| < d_3 < \frac{\sqrt 3}2 r+ \sqrt{ d_1^2 - \frac{r^2}4}$.

Thus, if $\frac r {\sqrt 3}<d_1\le r$ and $\left|\frac{\sqrt 3}2 r- \sqrt{ d_1^2 - \frac{r^2}4}\right| < d_3 \le d_1$, then $X=S_{12+}$  and if $\frac r {\sqrt 3}<d_1\le r$ and $d_3=d_1$, then $X=S_{123+}$. 
If $d_1>r$ and $\sqrt{d_1^2 - r^2} < d_3 < d_1$, then $X=S_{12+}$ and if $d_3=d_1>r$, then $X=S_{123+}$.
 
If $d_1 > r$ and $d_3 \in \left( \sqrt{d_1^2 - \frac{r^2}4} - \frac{\sqrt 3 r}2, \sqrt{d_1^2 - r^2} \right)$,  then $X\subset \{S_{12+}, S_{23-},S_{31-}\}$.
First, we have
$$ O(S_{12+}) = d_3^2 - \left( \sqrt{d_1^2 - \frac {r^2}4}   - \frac {\sqrt 3 r}2 \right)^2 = d_3^2 - d_1^2 - \frac {r^2}2 + \sqrt 3 r \sqrt{d_1^2 - \frac {r^2}4}. $$

Second, we can compute
\begin{eqnarray*}
O(S_{31-}) &=& \norm{S_{31-} - Z_2}^2 - d_1^2 = d_3^2 + r^2 - 2r d_3 \cos{\angle S_{31-} Z_3 Z_2}-d_1^2\\
                    &=& \norm{S_{31-} - Z_2}^2 - d_1^2 = d_3^2 + r^2 - 2r d_3 \cos{\angle S_{31-} Z_3 Z_1+ \angle Z_1 Z_3 Z_2}-d_1^2\\
                   &=&  d_3^2 -d_1^2+ r^2 - 2r d_3 \left(  \frac 1  2 \cos \angle S_{31-} Z_3 Z_1 - \frac {\sqrt 3} 2 \sin \angle S_{31-} Z_3 Z_1 \right)\\
                   &=&  d_3^2 -d_1^2 + r^2 - 2r d_3 \left(  \frac 1  2 \frac{d_3^2 + r^2 - d_1^2}{2r d_3} - \frac {\sqrt 3} 2 \frac{ \sqrt{4r^2 d_3^2 - ( d_3^2 + r^2 - d_1^2)^2  }}{2rd_3}  \right)\\
                   &=& \frac{d_3^2}2 - \frac{d_1^2}2 + \frac{r^2}2 + \frac{\sqrt 3}2 \sqrt{4r^2 d_3^2 - ( d_3^2 + r^2 - d_1^2)^2  }
\end{eqnarray*}

Hence, subtracting the above two terms:
\begin{equation*}
O(S_{31-})-O(S_{12+}) = -\frac{d_3^2}2 + \frac{\sqrt 3}2 \sqrt{4r^2 d_3^2 - ( d_3^2 + r^2 - d_1^2)^2  } + \frac{d_1^2}2 + r^2   -\sqrt 3 r \sqrt{d_1^2 - \frac {r^2}4}.
\end{equation*}
If we change the variable as $u=d_1^2 - d_3^2$, then we have $r^2<u<\sqrt 3 r \sqrt{d_1^2-\frac{r^2}4} - \frac {r^2}2$. Let us think  $O(S_{31-})-O(S_{12+})$ as the function $f$ depending on $u$ for given $r$ and $d_1$ as follows: 
$$ f(u) = \frac {\sqrt 3} 2 \sqrt{ 4r^2 d_1^2 - (u+r^2)^2} + \left[ \frac u  2  + r^2 - \sqrt 3 r \sqrt{d_1^2 - \frac {r^2} 4}\right]. $$ 
Note that the first part of the righthand side is the part of a circle centered at $(-r^2,0)$ with radius $2rd_1$, which is decreasing on the interval $u\in \left[r^2, \sqrt 3 r \sqrt{d_1^2-\frac{r^2}4} - \frac {r^2}2\right]$ and the last part of the right hand side in the bracket is a increasing straight line wirh respect to $u$. Thus, $f$ is a decreasing function on the interval.

Then, it is enough to show that $f\left(\sqrt 3 r \sqrt{d_1^2-\frac{r^2}4} - \frac {r^2}2\right)\ge 0$ to prove $X=S_{12+}$.  Let us show that $f(\sqrt 3 d_3 r + r^2)>0$:
\begin{eqnarray*}
f\left(\sqrt 3 r \sqrt{d_1^2-\frac{r^2}4} - \frac {r^2}2\right)
&=& \frac{\sqrt 3}2 \sqrt{ 4r^2 d_1^2 - \left(\sqrt 3 r \sqrt{d_1^2 -\frac{r^2}4} + \frac{r^2}2\right)^2 } + \frac{\sqrt 3}2 r \sqrt{d_1^2-\frac{r^2}4} + \frac{3r^2}4 - \sqrt 3 r\sqrt{d_1^2 - \frac{r^2} 4} \\
&=& \frac{\sqrt 3 r}2 \sqrt{ \left(d_1^2 -  \frac{r^2}4\right) - \sqrt 3 r \sqrt{d_1^2 -\frac{r^2}4} + \frac{3r^2}4} - \frac{\sqrt 3}2 r \sqrt{d_1^2-\frac{r^2}4} + \frac{3r^2}4 \\
&=&  0 
\end{eqnarray*} 
Note that $u=\sqrt 3 r \sqrt{d_1^2-\frac{r^2}4} - \frac {r^2}2$ implies that the three circles meet at one point, which could be another proof of $f\left(\sqrt 3 r \sqrt{d_1^2-\frac{r^2}4} - \frac {r^2}2\right)=0$. 
Hence, by summing up the above results we have proved that if $d_1\ge \frac r{\sqrt 3}$ and $d_3<d_1$, then $X=S_{12+}$ and if $d_3=d_1\ge \frac r{\sqrt 3}$, then $X=S_{123+}$. 

If $\frac r {\sqrt 3}< d_1 \le r$ and $d_1<d_3 < \sqrt{d_1^2+ \frac{r^2}2}$, then, $X= \{S_{23+},S_{31+}\}$ by \eqref{eq:equid10d30}.

If $\frac r {\sqrt 3}< d_1 \le r$ and $\sqrt{d_1^2+\frac{r^2}2}<d_3<\frac {\sqrt 3 r}2 + \sqrt{d_1^2-\frac{r^2}4}$, then $X\subset \{S_{23+},S_{31+},S_{12-}\}$ by \eqref{eq:equid10d30}. Since $\{S_{12-}, S_{23+},S_{31+}\}\subset R_{110}, \; Y_3 \notin R_{110}$ and $\norm{Y_3-S_{12-}}< \norm{Y_3-S_{23+}}$, from the fact that the minimum points on $R_{110}$ are the farthest points from $Y_3$, we have $X=\{S_{23+},S_{31+}\}$.

Consider the case $d_1\ge  r$ and $d_1 < d_3 \le \sqrt{d_1^2 + \frac{r^2}2}$. Then, we have $X=\{ S_{23+}, S_{31+}\}$ by \eqref{eq:equid10d30} also.  

If  $d_1>  r$ and $\sqrt{d_1^2 + \frac{r^2}2} < d_3 < \frac {\sqrt 3}2   r + \sqrt{d_1^2 - \frac {r^2}4}$; then, $X\subset \{S_{23+},S_{31+},S_{12-}\}$ by \eqref{eq:equid10d30}.  Note that $\{S_{12-}, S_{23+},S_{31+}\}\subset R_{110}, Y_3 \in R_{110},$ and the minimum points of $R_{110}$ are the farthest points from $Y_3$. In other words, $X$ is the points among $\{S_{12-}, S_{23+},S_{31+}\}$ which are farthest from $Y_3$. If we plot the circle centered at $Y_3$ with radius $\norm{Y_3-S_{12-}}$ (let us consider $d_1=r$ as the limiting case),  the circle meets $S_1$ and $S_2$ at one point except $S_{12-}$, respectively. 
Let this point be $M_1$ and $M_2$, respectively.  Define $M:=\norm{Z_3-M_1}=\norm{Z_3-M_2}$. Therefore, if $\sqrt{d_1^2 + \frac {r^2} 2} < M < \frac {\sqrt 3}2   r + \sqrt{d_1^2 - \frac {r^2}4}$ and $r<d_1<d_3<M$, then $X=\{S_{23+}, S_{31+}\}$; If $d_1>r, d_3=M$, then $X=\{S_{23+}, S_{31+}, S_{12-}\}$; If $d_1>r, M<d_3<    \frac {\sqrt 3}2   r + \sqrt{d_1^2 - \frac {r^2}4}$, then $X=S_{12-}$. Let us compute $M$ in detail by computing $\angle Z_3  Y_3 M_2$. Since $\norm{Z_2 - M_2} = \norm{Z_2 - S_{12-}}$ and $\norm{Y_3 - M_2}=\norm{Y_3 - S_{12-}}$, $\angle Z_2 Y_3 M_2= \angle Z_2 Y_3 S_{12-} = \frac {5\pi} 6$ and from this we have $\angle Z_3 Y_3 M_2 = \frac {2\pi} 3$. Since $\norm{M_2 Y_3} = \sqrt{d_1^2 - \frac{r^2}4} - \frac {\sqrt 3 r}2$, we have
\begin{eqnarray*}
 M^2 &=& 3r^2 + \norm{M_2 Y_3}^2 - 2\sqrt 3 r \norm{M_2 Y_3} \cos\angle Z_3 Y_3 M_2 \\
         &=&  3r^2 + \norm{M_2 Y_3}^2 + \sqrt 3 r \norm{M_2 Y_3} \\
         &=& d_1^2 + 2r^2
\end{eqnarray*}         

If  $d_1>  r$ and $d_3 > \frac {\sqrt 3}2   r + \sqrt{d_1^2 - \frac {r^2}4}$, then $X=X_{001}$ by Theorem 6(a) and (c) in \cite{Kwon}. Since $Y_3\notin X_{001}$ and $X_{001}$ are the closest points from $Y_3$, we have $X=S_{12-}$.

Hence, we proved all the cases of the lemma.
\end{proof}

 \begin{figure}
\begin{minipage}[t]{8cm}
\centerline{\epsfig{file=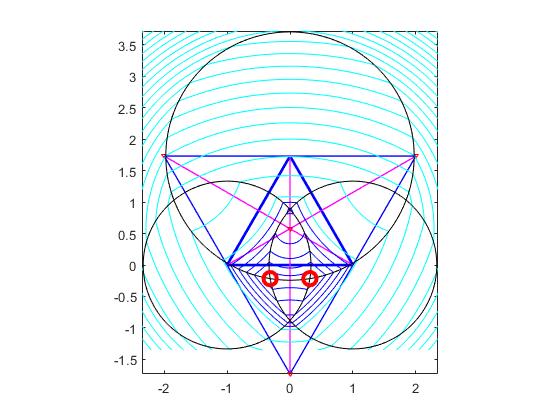, height=4cm,width=6cm,clip=1cm}}
\end{minipage}
\begin{minipage}[t]{8cm}
\centerline{\epsfig{file=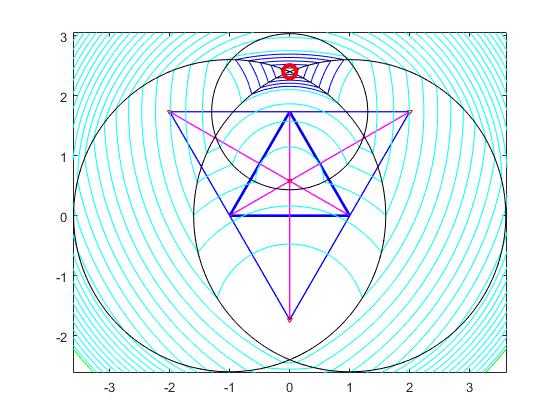, height=4cm,width=6cm,clip=1cm}}
\end{minipage}
\begin{center}
(a)\qquad\qquad\qquad\qquad\qquad\qquad\qquad\qquad\qquad\qquad(b)
\end{center}
\begin{minipage}[t]{8cm}
\centerline{\epsfig{file=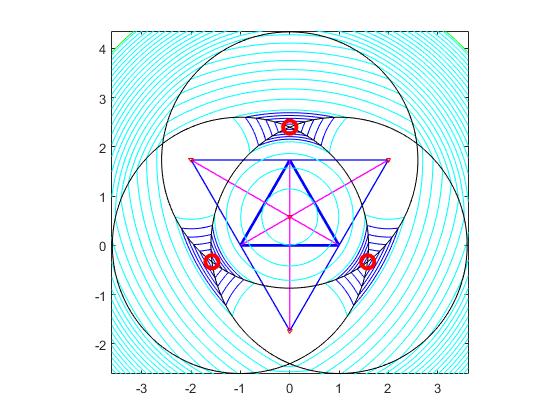, height=4cm,width=6cm,clip=1cm}}
\end{minipage}
\begin{minipage}[t]{8cm}
\centerline{\epsfig{file=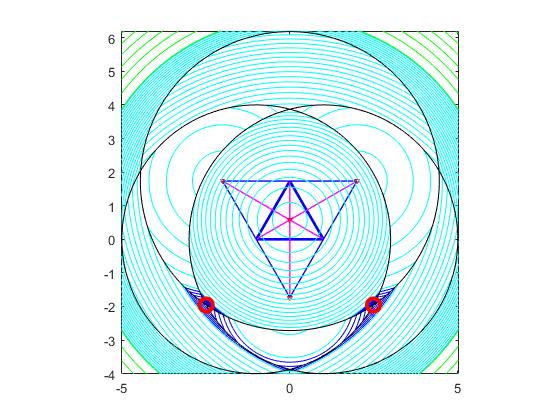, height=4cm,width=6cm,clip=1cm}}
\end{minipage}
\begin{center}
(c)\qquad\qquad\qquad\qquad\qquad\qquad\qquad\qquad\qquad\qquad(d)
\end{center}
\begin{minipage}[t]{8cm}
\centerline{\epsfig{file=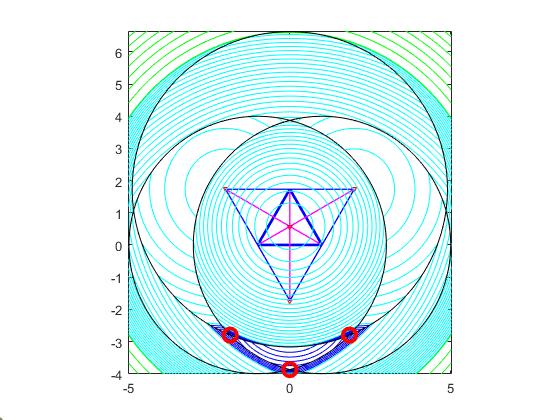, height=4cm,width=6cm,clip=1cm}}
\end{minipage}
\begin{minipage}[t]{8cm}
\centerline{\epsfig{file=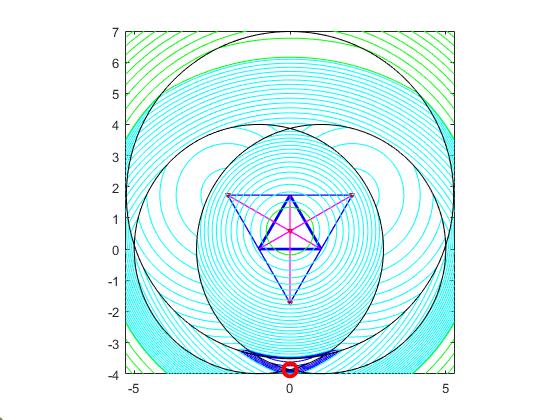, height=4cm,width=6cm,clip=1cm}}
\end{minipage}
\begin{center}
(e)\qquad\qquad\qquad\qquad\qquad\qquad\qquad\qquad\qquad\qquad(f)
\end{center}
\caption{Source locations when the measurement triangle is equilateral with $r=2$ (Theorem \ref{le:equid12})  and  when $(d_1, d_3) = $ (a) $(1.3333,1.9737)$ (b) $(2.6000,1.3000)$ (c) $(2.6000,2.6000)$ (d) $(4.0000,4.4495)$ (e)$(4.000,4.8990)$ (f)$(4.000,5.2520)$.}
\label{fig:lem11}
\end{figure}
\begin{table}
\begin{centering}
\begin{tabular}{|c||c|c|c|c|c|c|}
\hline
                                    &$O(S_{12+})$             &$O(S_{31+})$            &$O(S_{23+})$              &$O(S_{12-})$           &$O(S_{31-})$   &$O(S_{23-})$ \\
\hline \hline
Fig. \ref{fig:lem11} (a)   & 3.1726                          & {\color{red}1.2628}     &{\color{red} 1.2628}      & 2.9375                      & 7.3803              & 7.3803 \\
\hline
Fig. \ref{fig:lem11} (b)   & {\color{red}1.2438}     & 4.9420                        & 4.9420                          & 15.3838                    & 3.8720              & 3.8720\\
\hline
Fig. \ref{fig:lem11} (c)  &{\color{red} 6.3138}      &{\color{red} 6.3138}     & {\color{red} 6.3138}     & 10.3138                    & 10.3138            & 10.3138\\

\hline
Fig. \ref{fig:lem11} (d)  & 15.2144                       & {\color{red}9.9563}      &{\color{red}9.9563}       & 11.6184                    &17.7543             & 17.7543\\
\hline
Fig. \ref{fig:lem11} (e)  & 19.4164                       &{\color{red}7.4164}       &{\color{red}7.4164}       & {\color{red}7.4164}   & 19.4164            & 19.4164\\
\hline
Fig. \ref{fig:lem11} (f)  & 23.0000                       & 4.4093                          & 4.4093                         &{\color{red} 3.8328}    & 19.9929           &19.9929\\ 
\hline
\end{tabular}
\caption{The comparison of objective values at $S_{12\pm}, S_{31\pm},$ and $S_{23\pm}$ for the equilateral triangle with $r=2$. }
\label{tab:lem11}
\end{centering}
\end{table}

\section{When the measurement triangle is isosceles}   
 We will suppose that  the measurement triangle is an isosceles triangle with $\norm{Z_1 Z_3}=\norm{Z_2 Z_3}$ and $d_1=d_2$. Let $r:=\norm{Z_1 Z_2}$ and $s:=\norm{Z_3 - \frac{Z_1+Z_2}2}$.   
 Further, we will take into account the case that $d_1 \ge \sqrt{ \frac{r^2} 4 + \frac {s^2} 9}$ and $d_3 \in \left(\left|s-\sqrt{d_1^2 - \frac {r^2}4}\right|, s+\sqrt{d_1^2 - \frac {r^2}4}\right)$ for Lemma \ref{le:S12+S23+},\ref{le:S23+S12-},\ref{le:S12+S23-}, and \ref{le:isoscelesd12}. This implies that $Y_0\in D_1\cap D_2$ and 
all three disks meet, and $R_{i_1 i_2 i_3}$ has one connected component for every $i_1,i_2,i_3 \in \{0,1\}$ and thus, we have $X\subset S_{123}^\pm$ by Theorem \ref{th:main1}. 
\begin{lem}\label{le:S12+S23+}
Let $R$ be a  positive constant such that
$$ R^2:= d_1^2 -\frac{r^2}4 + s^2 \frac{4s^2 + r^2}{4s^2 + 9r^2}  - 2s \sqrt{d_1^2 - \frac{r^2}4} \frac{4s^2-3r^2}{4s^2+9r^2}, $$
then $R\in \left(\left|s-\sqrt{d_1^2 - \frac {r^2}4}\right|, s+\sqrt{d_1^2 - \frac {r^2}4}\right)$ and we have 
$$\begin{array}{ccc}
O(S_{31+})  > O(S_{12+}),   &{\mbox iff} & d_3<R,\\
O(S_{31+})  = O(S_{12+}),   &{\mbox iff} & d_3 = R,\\
O(S_{31+}) <  O(S_{12+}),   &{\mbox iff}&  d_3> R. 
\end{array}$$
Further, depending on $s/r$ and $d_1$, we have the following relations:
\begin{equation}\label{eq:R}
\begin{array}{ll}
{\mbox If}\; s= \frac{\sqrt 3}2r,&{\mbox then }\;  R=d_1.\\
&\\
{\mbox If}\;  s> \frac{\sqrt 3}2 r\;{\mbox and }&
\left\{\begin{array}{ccc}
d_1 > \sqrt{ \frac{r^2}4 + s^2 \left( \frac{4s^2 + 5r^2}{4s^2-3r^2}\right)^2},&{\mbox then }& \left|s-\sqrt{d_1^2 - \frac {r^2}4}\right|< R<\sqrt{d_1^2 - \frac{r^2} 4 - s^2},\\
d_1 = \sqrt{ \frac{r^2}4 + s^2 \left( \frac{4s^2 + 5r^2}{4s^2-3r^2}\right)^2}, &{\mbox then }& R = \sqrt{d_1^2 - \frac {r^2}4 - s^2},\\
d_1 < \sqrt{ \frac{r^2}4 + s^2 \left( \frac{4s^2 + 5r^2}{4s^2-3r^2}\right)^2}, &{\mbox then }& \sqrt{d_1^2 - \frac {r^2}4 - s^2} <R<\sqrt{d_1^2 - \frac {r^2} 4 + s^2}.
\end{array}\right.\\
&\\
{\mbox If}\;  s< \frac{\sqrt 3}2 r\;{\mbox and }&
\left\{\begin{array}{ccc}
d_1 > \sqrt{ \frac{r^2}4 + s^2 \left( \frac{4r^2}{3r^2-4s^2}\right)^2},&{\mbox then }& \sqrt{d_1^2 - \frac {r^2} 4 + s^2}<R<s+\sqrt{d_1^2 - \frac {r^2}4},\\
d_1 = \sqrt{ \frac{r^2}4 + s^2 \left( \frac{4r^2}{3r^2-4s^2}\right)^2},&{\mbox then }& R=\sqrt{d_1^2 - \frac {r^2} 4 + s^2},\\
d_1 < \sqrt{ \frac{r^2}4 + s^2 \left( \frac{4r^2}{3r^2-4s^2}\right)^2}, &{\mbox then }& \sqrt{d_1^2 - \frac {r^2}4 - s^2} <R<\sqrt{d_1^2 - \frac {r^2} 4 + s^2}.
\end{array}\right.
\end{array}
\end{equation}
\end{lem}
\begin{proof}
Let us compute $d_3=R$ satisfying $\norm{Y_0 - S_{12+}}=\norm{Y_0-S_{23+}}=\norm{Y_0-S_{31+}}$, which implies $O(S_{12+})=O(S_{23+})=O(S_{31+})$. By using Lemma \ref{le:2D+}
\begin{eqnarray*} 
R^2 &=&  \left(\frac 2 3 s\right)^2 + \left( \sqrt{d_1^2 - \frac{r^2}4} -\frac s 3\right)^2 -2 \frac 2 3 s \left( \sqrt{d_1^2 - \frac{r^2}4} -\frac s 3\right) \cos \angle Z_3 Y_0 S_{23+} \\
          &=& d_1^2 + \frac 5 9 s^2 - \frac{r^2} 4 - \frac 2 3 s\sqrt{d_1^2 - \frac {r^2} 4} - \frac 4 3 s \left( \sqrt{d_1^2 - \frac{r^2}4} -\frac s 3\right)\left( 2 \left(\frac{s/3}{\sqrt{\frac{r^2}4 +\frac{s^2}9}}\right)^2 -1\right)\\
              &=& d_1^2 -\frac{r^2}4 + s^2 \frac{4s^2 + r^2}{4s^2 + 9r^2}  - 2s \sqrt{d_1^2 - \frac{r^2}4} \frac{4s^2-3r^2}{4s^2+9r^2},
\end{eqnarray*}
which proves the lemma. Here, we have used  $\angle Z_3 Y_0 S_{23+} = \angle Z_1 Y_0 Z_2 = 2 \angle Z_1 Y_0 L $, which is from (\ref{eq:2D+}).  
The computation of \eqref{eq:R} is trivial.
\end{proof}

\begin{lem}\label{le:S23+S12-}
Suppose that $d_1\ge \sqrt{\frac{r^2}4 + s^2}$. 
Let $M$ be a positive number 
such that
$$ M^2:=   d_1^2 - \frac{r^2} 4  + 2s \sqrt{d_1^2 - \frac{r^2} 4} \frac{4s^2-3r^2}{4s^2 + r^2} + s^2 \frac{4s^2 + 9r^2}{4s^2 + r^2}. $$
Then,  $M\in \left(\sqrt{d_1^2 - \frac {r^2}4}-s, \sqrt{d_1^2 - \frac {r^2}4}+s\right)$ and we have
$$
\begin{array}{ccc}
O(S_{31+})  < O(S_{12-}),   &{\mbox iff} & d_3<M,\\
O(S_{31+})  = O(S_{12-}),   &{\mbox iff} & d_3= M,\\
O(S_{31+}) >  O(S_{12-}),   &{\mbox iff}&  d_3> M 
\end{array}
$$
And further, we have the following inequalities about $M$:
$$\begin{array}{ll}
{\mbox If}\;  s= \frac{\sqrt 3}2 r, &{\mbox then }\;  M= \sqrt{d_1^2 + 2r^2}.\\
&\\
{\mbox If}\;  s> \frac{\sqrt 3}2 r, &{\mbox then }\;  M>\sqrt{d_1^2 - \frac{r^2} 4 + s^2}.\\
&\\
{\mbox If}\;  s< \frac{\sqrt 3}2 r \;{\mbox and }&
\left\{\begin{array}{ccc}
 d_1 > \sqrt{ \frac{r^2}4 + s^2 \left( \frac{4r^2}{3r^2-4s^2}\right)^2},&{\mbox then }& M<\sqrt{d_1^2 - \frac {r^2} 4 + s^2},\\
d_1 = \sqrt{ \frac{r^2}4 + s^2 \left( \frac{4r^2}{3r^2-4s^2}\right)^2},&{\mbox then }& M=\sqrt{d_1^2 - \frac {r^2} 4 + s^2},\\
d_1 < \sqrt{ \frac{r^2}4 + s^2 \left( \frac{4r^2}{3r^2-4s^2}\right)^2}, &{\mbox then }& M>\sqrt{d_1^2 - \frac {r^2} 4 + s^2}.\\
\end{array}\right.
\end{array}
$$.
\end{lem}
\begin{proof}
In this case, we have $Y_3 \in R_{110}$. If we plot the circle centered at $Y_3$ with radius $\norm{Y_3-S_{12-}}$,  the circle meets $S_1$ and $S_2$ at one point except $S_{12-}$, respectively. Let this point be $M_1$ and $M_2$, respectively.  Define $M:=\norm{Z_3-M_1}=\norm{Z_3-M_2}$. 
Let us compute $M$ in detail. Since $\norm{Z_2 - M_2} = \norm{Z_2 - S_{12-}}$ and $\norm{Y_3 - M_2}=\norm{Y_3 - S_{12-}}$, we have
$$\angle Z_2 Y_3 M_2= \angle Z_2 Y_3 M_2 - \angle Z_3 Y_3 Z_2 = \angle Z_2 Y_3 S_{12-} -\angle Z_3 Y_3 Z_2 = \pi - 2\angle Z_3 Y_3 Z_2.$$
Since $\cos\angle Z_3 Y_3 Z_2 = \frac{s}{\sqrt{\frac{r^2}4 +s^2}}$ and $\norm{M_2 Y_3}= \sqrt{d_1^2 - \frac{r^2}4} - s$, we have
\begin{eqnarray*}
 M^2 &=& 4s^2 + \norm{M_2 Y_3}^2 - 4s \norm{M_2 Y_3} \cos\angle Z_3 Y_3 M_2 \\
         &=&  4s^2 + \norm{M_2 Y_3}^2 - 4s \norm{M_2 Y_3} (1-2\cos^2 \angle Z_3 Y_3 Z_2) \\
         &=& d_1^2 - \frac{r^2} 4  + 2s \sqrt{d_1^2 - \frac{r^2} 4} \frac{4s^2-3r^2}{4s^2 + r^2} + s^2 \frac{4s^2 + 9r^2}{4s^2 + r^2},
\end{eqnarray*}
which proved the lemma.
\end{proof}

Note that the constant $P= \sqrt{ \frac{r^2}4 + s^2 \left( \frac{4s^2 + 5r^2}{4s^2-3r^2}\right)^2}$ in Theorem \ref{th:main3} is only defined when $s>\frac{\sqrt 3 r}2$. 
\begin{lem}\label{le:S12+S23-}
Assume that  $d_1 > \sqrt{\frac{r^2}4 + s^2}$. 
Then, there is $d_3^*  \in \left(\sqrt{d_1^2 - \frac {r^2}4}-s, \sqrt{d_1^2 - \frac{r^2}4 - s^2}\right) $ when $s>\frac{\sqrt 3} 2 r  \mbox{ and } d_1 >  P$, such that 
\begin{subeqnarray}
{\mbox If}\; \left( s\le \frac{\sqrt 3} 2 r  \mbox{ or } d_1 <  P \right)&\mbox{ and }\; d_3 \in \left[\sqrt{d_1^2 - \frac {r^2}4}-s, \sqrt{d_1^2 - \frac{r^2}4 - s^2}\right] ,
\;{\mbox then}\;  O(S_{23-}) > O(S_{12+}) .\slabel{eq:S12+S23-1}\\
&\nonumber\\
{\mbox If}\; s>\frac{\sqrt 3} 2 r, d_1 =  P, \;{\mbox and}&
\left\{\begin{array}{ccc}
d_3 \in \left[\sqrt{d_1^2 - \frac {r^2}4}-s, \sqrt{d_1^2 - \frac{r^2}4 - s^2}\right),                 &{\mbox then}\;  O(S_{23-})  >  O(S_{12+}),\\
d_3 = \sqrt{d_1^2 - \frac{r^2}4 - s^2},                                                                                  &{\mbox then}\;  O(S_{23\pm}) = O(S_{12+}).\\
\end{array}\right.\slabel{eq:S12+S23-2}\\
&\nonumber\\
{\mbox If} s>\frac{\sqrt 3} 2 r,  d_1 >  P, \;{\mbox and}& 
\left\{\begin{array}{ccc}
d_3 \in \left[\sqrt{d_1^2 - \frac {r^2}4}-s, d_3^*\right),                                                         &{\mbox then}\;  O(S_{23-})  > O(S_{12+}) ,\\
d_3 = d_3^*,                                                                                                                          &{\mbox then}\;  O(S_{23-})  = O(S_{12+}) ,\\
d_3 \in \left[d_3^*, \sqrt{d_1^2 - \frac{r^2}4 - s^2}   \right],                                                 &{\mbox then}\;  O(S_{23-})  < O(S_{12+}). 
\end{array}\right.\slabel{eq:S12+S23-3}
\end{subeqnarray}
\end{lem}
\begin{proof}  
Using 
$$ \cos\angle S_{31-}  Z_3 Z_2 =\cos(\angle S_{31-} Z_3 Z_1+ 2\angle Z_1 Z_3 L) = \frac {s^2-\frac{r^2}4}{s^2+\frac{r^2}4}  \cos \angle S_{31-} Z_3 Z_1 - \frac {rs}{s^2+\frac{r^2}4} \sin \angle  S_{31-} Z_3 Z_ 1, $$ 
let us compute $O(S_{12+})$ and $O(S_{31-})$ as follows:
\begin{eqnarray*}
O(S_{12+}) &=& d_3^2 - \norm{ Z_3 - S_{12+}}^2 = d_3^2 - \left(\sqrt{d_1^2 - \frac{r^2}4} - s\right)^2 =  d_3^2 - \left(d_1^2 -\frac{r^2}4\right) - s^2 + 2s\sqrt{d_1^2 - \frac{r^2}4} 
\\
O(S_{31-})&=& -d_2^2 + \norm{Z_2 - S_{31-}}^2  = - d_1^2 + \left(  d_3^2 + |Z_1 Z_3|^2 - 2d_3 |Z_1 Z_3| \cos (\angle S_{31-} Z_3 Z_1 + \angle Z_1 Z_3 Z_2) \right)
\\
                   &=& d_3^2 + |Z_1 Z_3|^2 - d_1^2  -\frac 1 {|Z_1 Z_3|^2} \left( (d_3^2 + |Z_1 Z_3|^2 - d_1^2) (s^2 - \frac{r^2}4) - rs \sqrt{4d_3^2|Z_1 Z_3|^2 -  (d_3^2 + |Z_1 Z_3|^2 - d_1^2) ^2} \right)
\end{eqnarray*} 

Let $u =\sqrt{d_1^2 - \frac{r^2} 4} - s$. Then $u>0$. Next let us parametrize $d_3$ with respect to $t$ such as $d_3^2 = u^2 + 2tsu$. Then we have $0\le t \le 1, 0<u\le d_3\le u^2+2su,$ and $d_3^2 - d_1^2 + |Z_1 Z_3|^2 = 2(t-1)su$. 
Inserting these parametrizatons into $O(S_{12+})$ and $O(S_{31-})$, we have the following function for $t\in [0,1]$.
$$
g(t) = O(S_{31-}) - O(S_{12+}) = \frac{2s}{|Z_1 Z_3|^2} \left( -u\left[\left(s^2 - \frac{r^2}4\right)t + \frac{r^2}2\right]   + r\sqrt{ -s^2 u^2 t^2 + 2su(s^2 + \frac{r^2}4 + su)t + \frac{r^2 u^2}4 }\right)
$$
Note that $g(0)=0$ and the first function on the right hand side is a linear function and the second function is an increasing concave function. Hence,  we have $g''<0$ on $t\in[0,1]$.
Therefore, under the condition that $g(1)>0$, we have $g(t)>0$ iff $t\in(0,1]$ and $g(1)=0$ implies $g(t)>0$ if and only if $t\in (0,1)$. And $g(1)<0$ implies that there is a $t^*$ in $(0,1)$ such that 
$$\begin{array}{ccc}
    g(t) >0      &\mbox{is equivalent to} & t\in(0, t^*) \\
    g(t)<0        &\mbox{is equivalent to} & t\in(t^*,1)
\end{array}
$$    
Let $d_3^* = u^2 + 2t^* su, 0<t^*<1$. 

Let us find the condition on $u$ such that $g(1)=0$.  From
\begin{eqnarray*}
g(1) &=&  \frac{2s}{|Z_1 Z_3|^2} \left( -u\left(s^2 + \frac{r^2}4\right)  + r\sqrt{\left(s^2 + \frac{r^2}4\right)( u^2 + 2su)}\right)\\
      &=&   \frac{2sr\sqrt u}{|Z_1 Z_3|} \left(   \sqrt{u + 2s} -\frac{\sqrt{s^2 + \frac{r^2}4}}r \sqrt u\right)
\end{eqnarray*}
we have,
$$
\begin{array}{ccc}
g(1)>0 &\longleftrightarrow   &s\le \frac{\sqrt 3} 2 r  \mbox{ or } u <  Q    \\
g(1)=0 &\longleftrightarrow   &s> \frac{\sqrt 3} 2 r  \mbox{ and } u =  Q   \\
g(1)<0 &\longleftrightarrow   &s> \frac{\sqrt 3} 2 r  \mbox{ and } u >  Q
\end{array}
 $$
 where $Q:=\frac{2r^2s}{s^2- \frac 3 4 r^2}$. Since $u=Q$ implies $d_1 = P$  and vice versa, we have proved the lemma.  
 \end{proof} 
   
\begin{theorem}\label{le:isoscelesd12}
Suppose that $\norm{Z_1-Z_3}=\norm{Z_2-Z_3}$ and $d_1 = d_2$. Then, the source $X$ is determined as follows:
 
\begin{tabular}{|c|c|c||c|}
\hline
    s             &  $ d_1=d_2$                                                                                                 & $d_3$   											& $X$ \\
\hline \hline
$(0,\infty)$ &$\left(0,\sqrt{ \frac{r^2} 4 + \frac {s^2} 9} \right]$                                        & $\left(0, \frac{2s} 3\right]$  								& $Y_0$\\
\hline
                  &$\left(0,\sqrt{ \frac{r^2} 4 + s^2}\right]$                                  		         & $\left[ 2s, \infty \right)$ 									&$Y_3$\\
\hline
                  &$\left(0, \frac r 2\right]$                                                           		         & $\left[\frac {2s}3 , 2s \right]$ 								& $N_3$\\
\hline
                  &$\left(\frac r 2, \sqrt{ \frac{r^2} 4 + \frac {s^2} 9}\right]$        		         & $\left[\frac {2s}3, s- \sqrt{ d_1^2 - \frac{r^2}4 }\right]$  				& $N_3$\\
\hline
                  &$\left(\frac r 2, \sqrt{ \frac{r^2} 4 + s^2}\right]$                     		         & $\left[ s+ \sqrt{ d_1^2 - \frac{r^2}4},   2s \right]$ 					&$N_3$\\    
\hline
                  &$\left(\frac r 2, \sqrt{ \frac{r^2} 4 + \frac {s^2} 9}\right)$        		         & $\left(s - \sqrt{ d_1^2 - \frac{r^2}4}, s + \sqrt{ d_1^2 - \frac{r^2}4 }\right)$          & $\{S_{23+}, S_{31+}\}$  \\
\hline
                  &$\left[\sqrt{ \frac{r^2} 4 + \frac {s^2} 9},\infty\right)$                        	         & $\left(0, \left|s - \sqrt{d_1^2 - \frac{r^2}4}\right| \right]$ 	                      	& $S_{12+}$\\
\hline    
                  &$\left[\sqrt{ \frac{r^2} 4 + \frac {s^2} 9},\sqrt{ \frac{r^2} 4 + s^2}\right)$   & $\left[s - \sqrt{d_1^2 - \frac{r^2}4}, R \right)$                                                        & $S_{12+}$\\
\hline
                  &                                                                                                                     & $R$                                                                                                                        &$\{S_{12+}, S_{23+}, S_{31+}\}$\\
\hline
                  &                                                                                                                     & $\left(R, s + \sqrt{d_1^2 - \frac {r^2} 4}\right]$                                                     &$\{S_{23+},S_{31+}\}$\\                     
\hline                                                                                  
$\left(0, \frac{\sqrt 3 r}2\right)$     & $\left[ \sqrt{ \frac{r^2} 4 + s^2},\sqrt{\frac{r^2}4 + s^2 \left(\frac{4r^2}{3r^2-4s^2}\right)^2}\right)$   & $\left[ \sqrt{d_1^2 - \frac{r^2}4} - s, R\right)$                                                                                &$ S_{12+}$\\ 
\hline
                                                      &                                                                                                                                                                & $R$                                                                                                                                                &$S_{123}^+$\\
\hline
                                                      &                                                                                                                                                                & $\left(R,M \right)$                                                                                                                           &$\{S_{23+},S_{31+}\}$\\
\hline
                                                      &                                                                                                                                                                &$ M$      						           			                                        &$\{ S_{23+}, S_{31+}, S_{12-} \}$\\
\hline 
                                                      &      							                                                                                        &$\left( M, \infty\right) $  	                                                                                                                   &$S_{12-}$\\                                                                                               
\hline
                                                      & $\sqrt{\frac{r^2}4 + s^2 \left(\frac{4r^2}{3r^2-4s^2}\right)^2}$                                                             &$\left[ \sqrt{d_1^2 - \frac{r^2}4}-s,\sqrt{d_1^2 - \frac{r^2}4 + s^2}\right)$                                     &$S_{12+}$\\    
\hline
                                                      &                                                                                                                                                                &$\sqrt{d_1^2 - \frac{r^2}4 + s^2}$                                                                                                  &$\{S_{12\pm},S_{23+},S_{31+}\}$\\ 
\hline
                                                      &                                                                                                                                                                &$\left( \sqrt{d_1^2 - \frac{r^2}4 + s^2},\infty \right)$                                                                      &$S_{12-}$\\ 
\hline                                                      
                                                      & $\left(\sqrt{\frac{r^2}4 + s^2 \left(\frac{4r^2}{3r^2-4s^2}\right)^2},\infty \right)$                                  &$\left[ \sqrt{d_1^2 - \frac{r^2}4}-s,\sqrt{d_1^2 - \frac{r^2}4 + s^2}\right)$                                    &$S_{12+}$\\                                                       
\hline 
                                                      &                                                                                                                                                                & $\sqrt{d_1^2 - \frac{r^2}4 + s^2}$                                                                                                 &$S_{12\pm}$\\                                                       
\hline 
                                                      &                                                                                                                                                                &$\left( \sqrt{d_1^2 - \frac{r^2}4 + s^2},\infty \right)$                                                                      &$S_{12-}$\\ 
\hline                                                      
\end{tabular}

\begin{tabular}{|c|c|c||c|}
\hline
    s                                                  &  $ d_1=d_2$                                                                                                                                                               & $d_3$   												& $X$ \\
\hline \hline
$\left( \frac{\sqrt 3r}2,\infty\right)$ & $\left[ \sqrt{ \frac{r^2} 4 + s^2}, \sqrt{\frac{r^2}4 + s^2 \left(\frac{4s^2+5r^2}{4s^2-3r^2}\right)^2} \right)$           & $\left[ \sqrt{d_1^2 - \frac{r^2}4} - s, R \right)$                                                             &$S_{12+}$\\ 
\hline
                                                       &                                                                                                                                                                                    & $R$                                                                                                                              &$S_{123}^+$\\
\hline
                                                       &                                                                                                                                                                                    &$\left[ R, \sqrt{d_1^2 - \frac{r^2}4 - s^2} \right)$                                                           &$\{S_{23+},S_{31+}\}$\\
\hline                                                        
                                                      & $\sqrt{\frac{r^2}4 + s^2 \left(\frac{4s^2+5r^2}{4s^2-3r^2}\right)^2}$                                                                        & $\left[ \sqrt{d_1^2 - \frac{r^2}4} - s, \sqrt{d_1^2 - \frac{r^2}4 - s^2} \right)$            &$S_{12+}$\\                                                                    
\hline                                                      
                                                      &                                                                                                                                                                                     & $\sqrt{d_1^2 - \frac{r^2}4 - s^2}$                                                                                &$\{S_{12+},S_{23\pm},S_{31\pm}\}$\\  
\hline
                                                      &                                                                                                                                                                                    & $\left( \sqrt{d_1^2 - \frac{r^2}4 - s^2},\sqrt{d_1^2 - \frac{r^2}4 + s^2}\right]$           &$\{S_{23+},S_{31+}\}$\\     
\hline
                                                    &    $\left(\sqrt{\frac{r^2}4 + s^2 \left(\frac{4s^2+5r^2}{4s^2-3r^2}\right)^2}, \infty\right)$                                         & $\left( \sqrt{d_1^2 - \frac{r^2}4} - s,d_3^* \right)$                                                        &  $ S_{12+}$\\
\hline
                                                    &                                              			                                                                                                         & $d_3^*$ 					                      			                    & $\{ S_{12+}, S_{23-}, S_{31-} \}$\\
\hline
		                               &							                                                                                                         & $\left( d_3^*,    \sqrt{d_1^2 -\frac{r^2}4 + s^2} \right)$                                                 & $\{S_{23-}, S_{31-}\}$\\              
\hline
                                                    &                                                                                                                                                                                   & $ \sqrt{d_1^2 - \frac{r^2}4 + s^2}$                                                                                & $\{ S_{23\pm}, S_{31\pm} \}$\\                                  
\hline
                                                     &$\left[ \sqrt{ \frac{r^2} 4 + s^2},  \infty \right)$                                                                                                         & $\left( \sqrt{d_1^2 -\frac{r^2}4 + s^2},  M  \right)$ 		            	                  & $\{ S_{23+}, S_{31+} \}$\\
\hline
                                                     &                                                                                                                                                                                 & $ M$      						           			                  & $\{ S_{23+}, S_{31+}, S_{12-} \}$\\
\hline 
                                                    &      							                                                                                                        & $\left( M, \infty\right) $  	                                                                                             & $S_{12-}$\\                               
\hline                                                                                                                                                                                                                                                                                             
\end{tabular}
\end{theorem}
\begin{proof}
The proofs for the cases $X=Y_0, X=Y_3,$ and $X=N_3$ are similar to that of Lemma \ref{le:equid12}. 

If $d_1\ge \sqrt{\frac {r^2}4 + \frac{s^2}9}$ and $d_3\le \left|s-\sqrt{d_1^2 - \frac{r^2}4}\right|$ , then $X=S_{12+}$ by similar argument as Lemma \ref{le:equid12}.

If $d_1\ge \sqrt{\frac {r^2}4 + s^2}$ and $d_3\ge s+\sqrt{d_1^2 - \frac{r^2}4}$ , then $X=S_{12-}$ also by similar argument as Lemma \ref{le:equid12}.

If $d_1\in \left( \frac r 2, \sqrt{ \frac{r^2} 4 + \frac {s^2} 9} \right)$ and $d_3 \in \left( s-\sqrt{ d_1^2 - \frac{r^2}4}, s+ \sqrt{ d_1^2 - \frac{r^2}4}\right)$, then 
all three disks meet and $R_{i_1 i_2 i_3}$ has one connected component for every $i_1,i_2,i_3 \in \{0,1\}$  and  we have $X\ge S_{123}^\pm$ by Theorem \ref{th:main1}. 
Using Lemma \ref{le:d10d30} and the similar argument as Lemma \ref{le:equid12}, we have
$X=\{S_{23+}, S_{31+}\}$. 

The remaining cases are when $d_1 \ge \sqrt{ \frac{r^2} 4 + \frac {s^2} 9}$ and $d_3 \in \left(\left|s-\sqrt{d_1^2 - \frac {r^2}4}\right|, s+\sqrt{d_1^2 - \frac {r^2}4}\right)$. 
In this all three disks meet and $R_{i_1 i_2 i_3}$ has one connected component for every $i_1,i_2,i_3 \in \{0,1\}$  and  we have $X\ge S_{123}^\pm$ by Theorem \ref{th:main1}. 
Using Lemma \ref{le:d10d30}, we have
\begin{subeqnarray}\label{eq:SAclassify}
&\mbox{If }& \sqrt{ \frac{r^2} 4 + \frac {s^2} 9} \le d_1 < \sqrt{\frac{r^2}4 + s^2} \mbox{ and } d_3 \in \left(s-\sqrt{d_1^2 - \frac {r^2}4}, \sqrt{d_1^2 - \frac{r^2}4 + s^2}\right), \nonumber\\
&& \mbox{ then } X \subset \{ S_{12+}, S_{23+}, S_{31+}  \}. \slabel{eq:SAclassify1}\\
&\mbox{If }& d_1 > \sqrt{\frac{r^2}4 + s^2} \mbox{ and } d_3 \in \left(\sqrt{d_1^2 - \frac {r^2}4}-s, \sqrt{d_1^2 - \frac{r^2}4 - s^2}\right), 
\mbox{ then } X\subset  \{S_{12+}, S_{23-}, S_{31-} \}.\slabel{eq:SAclassify2}\\
&\mbox{If }& d_1 > \sqrt{\frac{r^2}4 + s^2} \mbox{ and } d_3 =  \sqrt{d_1^2 - \frac{r^2}4 - s^2}, 
\mbox{ then } X\subset  \{S_{12+}, S_{23\pm}, S_{31\pm} \}.\slabel{eq:SAclassify3}\\
&\mbox{If }& d_1 > \sqrt{\frac{r^2}4 + s^2} \mbox{ and } d_3 \in \left(\sqrt{d_1^2 - \frac{r^2}4 - s^2}, \sqrt{d_1^2 - \frac{r^2}4 + s^2} \right), 
\mbox{ then } X\subset  \{S_{12+}, S_{23+}, S_{31+} \}.\slabel{eq:SAclassify4}\\
&\mbox{If }& d_3 = \sqrt{d_1^2 - \frac{r^2}4 + s^2},
\mbox{ then } X\subset  \{S_{12\pm}, S_{23+}, S_{31+} \}.\slabel{eq:SAclassify5}\\
&\mbox{If }& d_3 \in \left(\sqrt{d_1^2 - \frac{r^2}4 + s^2}, s+ \sqrt{d_1^2-\frac {r^2}4}\right),
\mbox{ then } X\subset  \{S_{23+}, S_{31+}, S_{12-} \}.\slabel{eq:SAclassify6}
\end{subeqnarray}

When $d_1\in \left[\sqrt{ \frac{r^2} 4 + \frac {s^2} 9},\sqrt{ \frac{r^2} 4 + s^2}\right) $ and $d_3 \in \left( \sqrt{d_1^2 - \frac{r^2}4 + s^2}, s + \sqrt{d_1^2 - \frac{r^2}4}\right)$, 
by using \eqref{eq:SAclassify6} we have $X\ge X_{001}$. Since the minimum points are the farthest points from $Y_3$ and $Y_3$ is not contained in $R_{001}$, we have $X=X_{001}=\{S_{23+},S_{31+}\}$.
Note that in case  \eqref{eq:SAclassify1}, the constant $R$ in Lemma \ref{le:S12+S23+} satisfies
$$ s - \sqrt{d_1^2 - \frac{r^2}4} < R < \sqrt{d_1^2 - \frac{r^2}4 + s^2} . $$
From the above and by using \eqref{eq:SAclassify5} and Lemma \ref{le:S12+S23+}, we can prove all the case for $d_1\in \left[\sqrt{ \frac{r^2} 4 + \frac {s^2} 9},\sqrt{ \frac{r^2} 4 + s^2}\right) $ and $d_3 \in \left( \sqrt{d_1^2 - \frac{r^2}4 + s^2}, s + \sqrt{d_1^2 - \frac{r^2}4}\right)$.

Thus far, we have proved all the cases except $d_1\ge \sqrt{\frac{r^2}4+s^2}$ and $d_3 \in \left[ \sqrt{d_1^2 - \frac{r^2}4} -s, \sqrt{d_1^2 - \frac{r^2}4 + s^2}\right]$. In the case, we should divide the original case into two cases: 
the flat isosceles ($s<\frac{\sqrt 3 r}2$) and the sharp isosceles ($s>\frac{\sqrt 3 r}2$) .

Let us consider the flat isosceles with $s< \frac{\sqrt 3} 2 r$. In the case in \eqref{eq:SAclassify2}, we have $X\subset \{S_{12+},S_{23-},S_{31-}\}$. 
Using \eqref{eq:S12+S23-1}, we have $X=S_{12+}$.
In the other cases, by applying Lemma \ref{le:S12+S23+}, Lemma \ref{le:S23+S12-}, and \eqref{eq:SAclassify3}-\eqref{eq:SAclassify6}, we can prove all the cases for $s<\frac{\sqrt 3}2r$.

Then, let us consider the sharp isosceles with $s> \frac{\sqrt 3}2r$. The case $d_1\ge \sqrt{\frac{r^2}4 + s^2}$ and $d_3\ge \sqrt{d_1^2 - \frac{r^2}4 + s^2}$ can be proved with Lemma \ref{le:S23+S12-}, \eqref{eq:SAclassify5}, and \eqref{eq:SAclassify6}.
In the other cases, by applying Lemma \ref{le:S12+S23+}, Lemma \ref{le:S23+S12-}, \eqref{eq:S12+S23-2},\eqref{eq:S12+S23-3}, and \eqref{eq:SAclassify3}-\eqref{eq:SAclassify6}, we can prove all the cases for $s<\frac{\sqrt 3}2r$.

Hence, we proved all the cases of the theorem.
\end{proof}     

In Fig. \ref{fig:lem15ss}, \ref{fig:lem15sl1}, and \ref{fig:lem15sl3},  interesting examples when $R_{i1,i2,i3}$ are all nonempty and connected for $i_1,i_2,i_3\in \{0,1\}$ are shown in Theorem {le:isoscelesd12}. The objective values for the cases shown in the figures are displayed in 
Table \ref{tab:th15}, in which red ones in each row represent the objective values for the solutions in each case.
\begin{table}
\begin{centering}
\begin{tabular}{|c||c|c|c|c|c|c|}
\hline
                                        &$O(S_{12+})$             &$O(S_{31+})$            &$O(S_{23+})$              &$O(S_{12-})$           &$O(S_{31-})$    &$O(S_{23-})$ \\
\hline \hline
Fig. \ref{fig:lem15ss} (a)  &{\color{red}2.8643}       &{\color{red}2.8643}     &{\color{red}2.8643}       &3.2429                       &6.4858                &6.4858\\
\hline
Fig. \ref{fig:lem15ss} (b)  &3.4103                          &{\color{red}2.5370}      &{\color{red}2.5370}       &2.6969                       &7.2504                &7.2504\\
\hline
Fig. \ref{fig:lem15ss} (c)  &{\color{red}0.5567}       & 4.6636                        &4.6636                           &7.4433                       &1.7770                &1.7770\\
\hline
Fig. \ref{fig:lem15ss} (d)  &{\color{red}4.0000}      &{\color{red}4.0000}      &{\color{red}4.0000}        &{\color{red}4.0000}   &8.0000                  &8.0000\\
\hline
Fig. \ref{fig:lem15ss} (e)  &{\color{red}4.0491}       &13.4171                        &13.4171                         &13.3865                     &8.0797                  &8.0797\\ 
\hline
Fig. \ref{fig:lem15ss} (f) &{\color{red}8.7178}       &10.4900                        &10.4900                         &{\color{red}8.7178}   &14.4900                   &14.4900\\
\hline
Fig. \ref{fig:lem15sl1} (a)  &{\color{red}6.5105}     &12.9986                        &12.9986                        &53.7055                     &10.7596                    &10.7596\\
\hline
Fig. \ref{fig:lem15sl1} (b)  &{\color{red}16.0720}    &{\color{red}16.0720}    &{\color{red}16.0720}     &44.1440                    &17.6576                     &17.6576\\
\hline
Fig. \ref{fig:lem15sl1} (c)  &22.6289                        &{\color{red}16.4606}    &{\color{red}16.4606}     &37.5872                    &20.6689                     &20.6689\\
\hline
Fig. \ref{fig:lem15sl1} (d)  &{\color{red}10.6491}    & 20.5469                       &20.5469                         &73.3509                    &15.2066                     &15.2066 \\
\hline
Fig. \ref{fig:lem15sl1} (e)  &{\color{red}24.0000}    &{\color{red}24.0000}    &{\color{red}24.0000}      &60.0000                    &{\color{red}24.0000}   &{\color{red}24.0000}\\
\hline
Fig. \ref{fig:lem15sl1} (f)  &32.5832                         &{\color{red}24.2271}   &{\color{red}24.2271}      &51.4168                     &27.6604                       &27.6604\\ 
\hline
Fig. \ref{fig:lem15sl3} (a)  &{\color{red}28.6572}    &36.0460                       &36.0460                         &94.5497                    &30.0675                        &30.0675\\
\hline
Fig. \ref{fig:lem15sl3} (b)  &{\color{red}31.8414}   &36.5462                        &36.5462                         &91.3655                    &{\color{red}31.8414}    &{\color{red}31.8414}\\
\hline 
Fig. \ref{fig:lem15sl3} (c)  &42.4272                       &37.2617                        &37.2617                         &80.7797                    &{\color{red}36.7912}  &{\color{red}36.7912}\\
\hline
\end{tabular}       
\caption{The comparison of objective values at $S_{12\pm}, S_{31\pm},$ and $S_{23\pm}$  for Theorem \ref{le:isoscelesd12}. The red numbers are the minimum objective values.
}
\label{tab:th15}
\end{centering}
\end{table}
 \begin{figure}
\begin{minipage}[t]{8cm}
\centerline{\epsfig{file=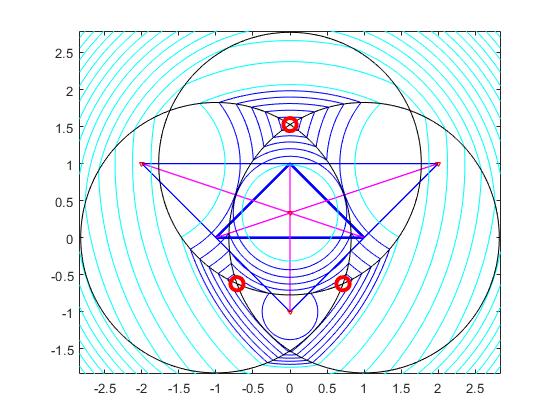, height=4cm,width=6cm,clip=1cm}}
\end{minipage}
\begin{minipage}[t]{8cm}
\centerline{\epsfig{file=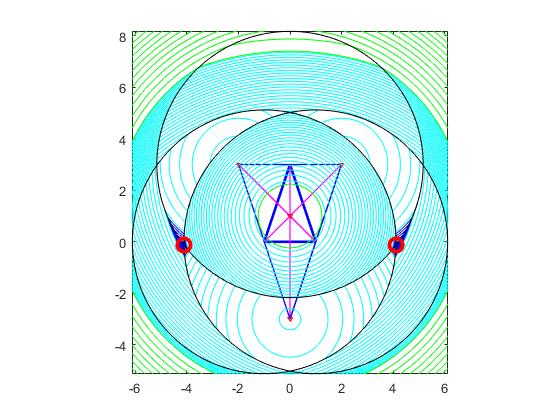, height=4cm,width=6cm,clip=1cm}}
\end{minipage}
\begin{center}
(a)\qquad\qquad\qquad\qquad\qquad\qquad\qquad\qquad\qquad\qquad(b)
\end{center}
\begin{minipage}[t]{8cm}
\centerline{\epsfig{file=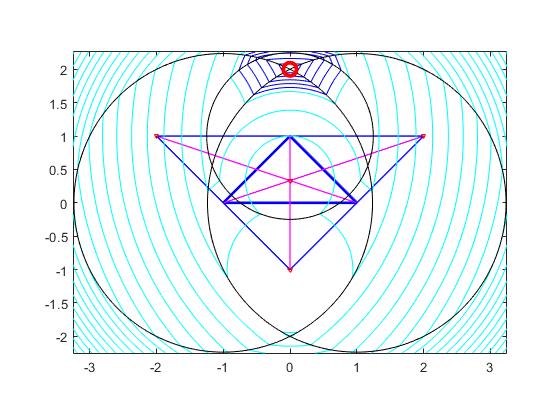, height=4cm,width=6cm,clip=1cm}}
\end{minipage}
\begin{minipage}[t]{8cm}
\centerline{\epsfig{file=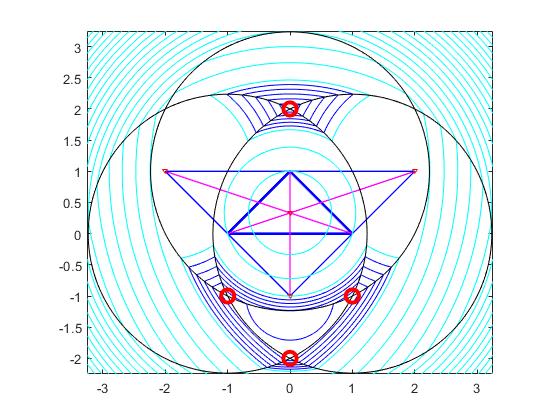, height=4cm,width=6cm,clip=1cm}}
\end{minipage}
\begin{center}
(c)\qquad\qquad\qquad\qquad\qquad\qquad\qquad\qquad\qquad\qquad(d)
\end{center}
\begin{minipage}[t]{8cm}
\centerline{\epsfig{file=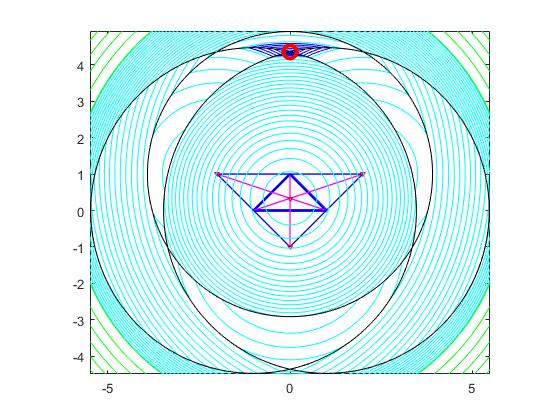, height=4cm,width=6cm,clip=1cm}}
\end{minipage}
\begin{minipage}[t]{8cm}
\centerline{\epsfig{file=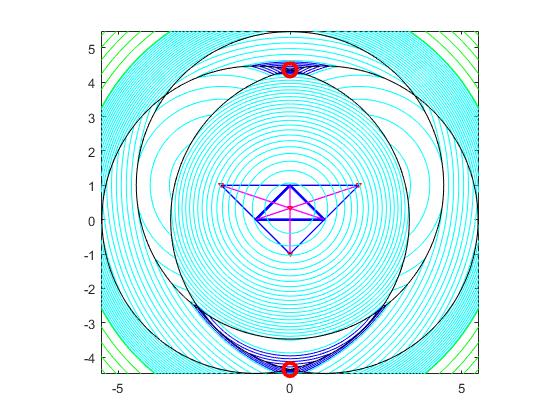, height=4cm,width=6cm,clip=1cm}}
\end{minipage}
\begin{center}
(e)\qquad\qquad\qquad\qquad\qquad\qquad\qquad\qquad\qquad\qquad(f)
\end{center}
\caption{Solution positions for flat isoscleles $ (s<\frac{\sqrt 3 r}2)$ with $s=1$ and $r=2$ in Theorem \ref{le:isoscelesd12}, when $(d_1, d_3)= $  (a) $(1.8251,1.7725)$ 
(b) $(1.8251,1.9204)$ (c) $(2.2361, 1.2477)$ (d) $(2.2361,2.2361)$ (e) $(4.4721,3.9155)$ (f) $(4.4721,4.4721)$.}
\label{fig:lem15ss}
\end{figure}

 \begin{figure}
\begin{minipage}[t]{8cm}
\centerline{\epsfig{file=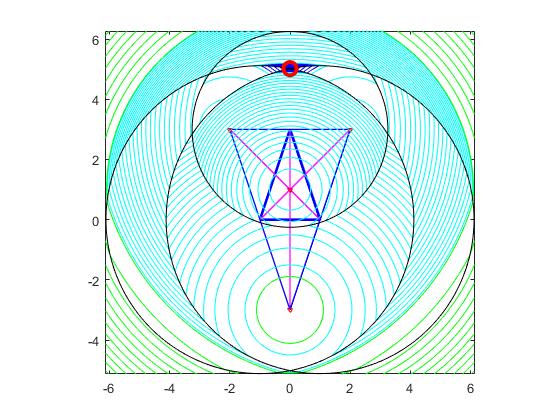, height=4cm,width=6cm,clip=1cm}}
\end{minipage}
\begin{minipage}[t]{8cm}
\centerline{\epsfig{file=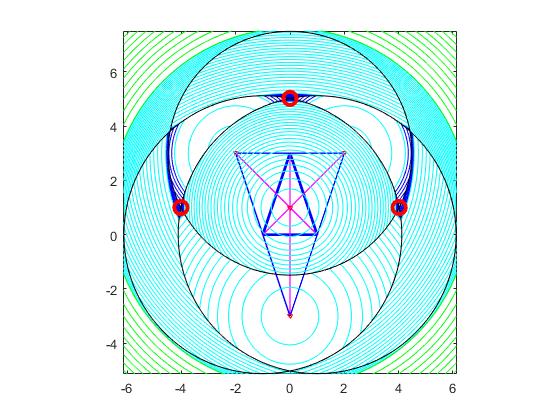, height=4cm,width=6cm,clip=1cm}}
\end{minipage}
\begin{center}
(a)\qquad\qquad\qquad\qquad\qquad\qquad\qquad\qquad\qquad\qquad(b)
\end{center}
\begin{minipage}[t]{8cm}
\centerline{\epsfig{file=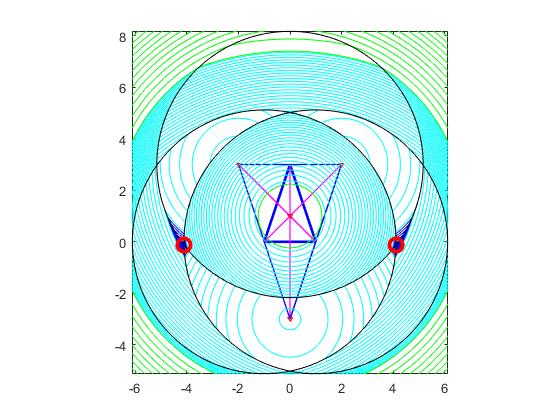, height=4cm,width=6cm,clip=1cm}}
\end{minipage}
\begin{minipage}[t]{8cm}
\centerline{\epsfig{file=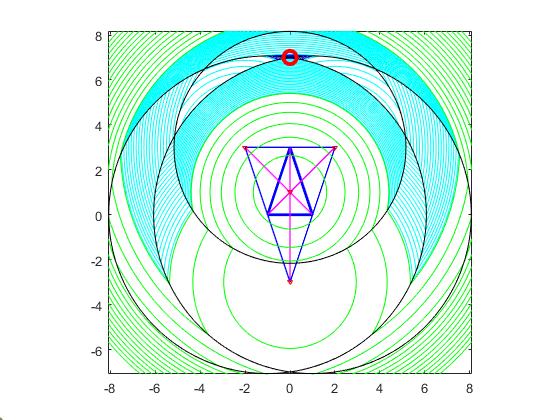, height=4cm,width=6cm,clip=1cm}}
\end{minipage}
\begin{center}
(c)\qquad\qquad\qquad\qquad\qquad\qquad\qquad\qquad\qquad\qquad(d)
\end{center}
\begin{minipage}[t]{8cm}
\centerline{\epsfig{file=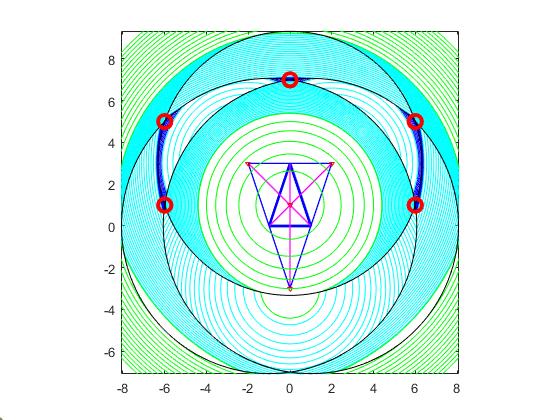, height=4cm,width=6cm,clip=1cm}}
\end{minipage}
\begin{minipage}[t]{8cm}
\centerline{\epsfig{file=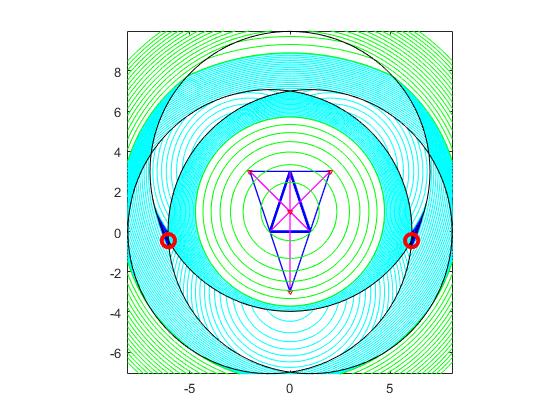, height=4cm,width=6cm,clip=1cm}}
\end{minipage}
\begin{center}
(e)\qquad\qquad\qquad\qquad\qquad\qquad\qquad\qquad\qquad\qquad(f)
\end{center}
\caption{Solution positions for sharp isosceles $(s>\frac{\sqrt 3 r}2)$ with $s=3,\; r=2,$  and  $\sqrt{\frac{r^2}4 +s^2}\le d_1\le \sqrt{\frac{r^2}4 + s^2 \left(\frac{4s^2+5r^2}{4s^2-3r^2}\right)^2}$ (Theorem \ref{le:isoscelesd12}) when $(d_1,d_3)= $ (a) $(5.1167,3.2531)$ (b) $(5.1167,4.4882)$ (c) $( 5.1167,5.1673)$ (d) $(7.0711,5.1623)$ (e) $(7.0711,6.3246)$ (f) $(7.0711, 6.9702)$.}
\label{fig:lem15sl1}
\end{figure}
 \begin{figure}
\begin{minipage}[t]{8cm}
\centerline{\epsfig{file=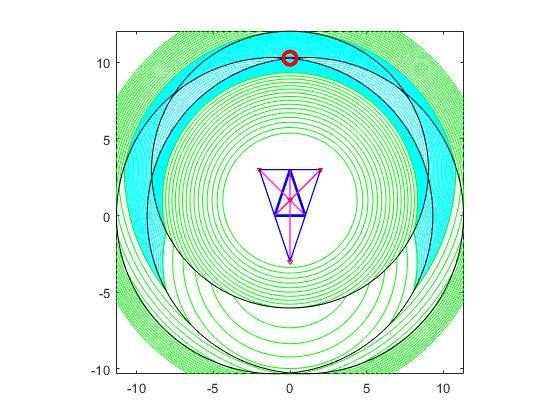, height=4cm,width=6cm,clip=1cm}}
\end{minipage}
\begin{minipage}[t]{8cm}
\centerline{\epsfig{file=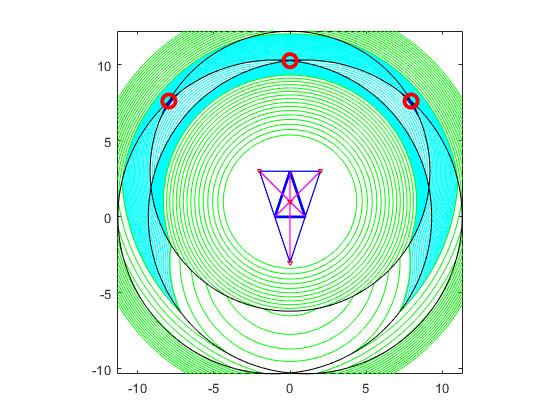, height=4cm,width=6cm,clip=1cm}}
\end{minipage}
\begin{center}
(a)\qquad\qquad\qquad\qquad\qquad\qquad\qquad\qquad\qquad\qquad(b)
\end{center}
\begin{minipage}[t]{8cm}
\centerline{\epsfig{file=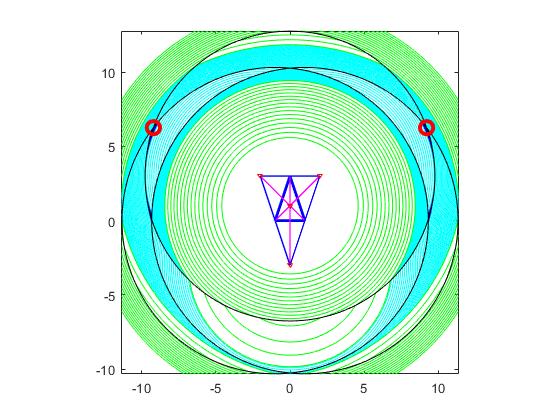, height=4cm,width=6cm,clip=1cm}}
\end{minipage}
\begin{minipage}[t]{8cm}
\centerline{\epsfig{file=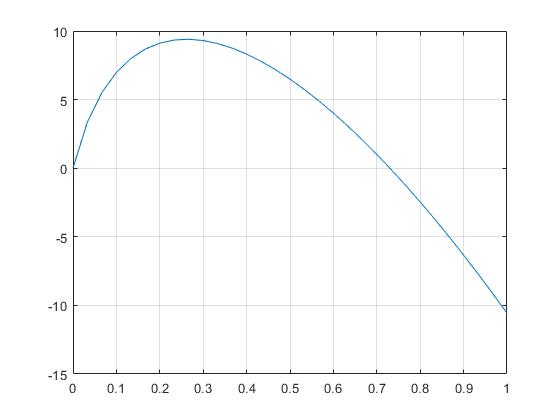, height=4cm,width=6cm,clip=1cm}}
\end{minipage}
\begin{center}
(c)\qquad\qquad\qquad\qquad\qquad\qquad\qquad\qquad\qquad\qquad(d)
\end{center}
\caption{Solution positions for sharp isosceles $(s>\frac{\sqrt 3 r}2)$ with $s=3, r=2,\;\sqrt{\frac{r^2}4 + s^2 \left(\frac{4s^2+5r^2}{4s^2-3r^2}\right)^2}<d_1=10.3158\le P$ (Theorem \ref{le:isoscelesd12}) and $d_3^*=9.2008$ when (a) $d_3=9.0261<d_3^*$ (b) $d_3=d_3^*$ (c) $d_3=9.7591 >d_3^*$. (d) The graph for $f(t)$ in Lemma \ref{le:S12+S23-}, in which the $x$-intercept point of $f(t)$ is $t^*=0.7302 $ (up to 4 decimals).}
\label{fig:lem15sl3}
\end{figure}


\begin{corollary}\label{le:isoNumber}
Suppose that $\norm{Z_1-Z_3}=\norm{Z_2-Z_3}$ and $d_1 = d_2$. Then, we have the followings: 
\begin{itemize}
\item{ $X=\{S_{12+},S_{23\pm},S_{31\pm}\}$ iff $\left[s> \frac{\sqrt 3}2 r, d_1 = \sqrt{\frac{r^2}4 +s^2\left(\frac{4s^2+5r^2}{4s^2-3r^2} \right)^2},   d_3 = \frac{4sr}{4s^2-3r^2}\sqrt{4s^2+r^2}\right]$.}
\item{ $X=\{S_{23\pm}, S_{31\pm} \}$ iff $\left[ s>  \frac{\sqrt 3}2 r, d_1>\sqrt{\frac{r^2}4 +s^2\left(\frac{4s^2+5r^2}{4s^2-3r^2} \right)^2},  d_3 = \sqrt{d_1^2 - \frac {r^2}4 - s^2}\right]$.}
\item{ $X=\{S_{12\pm}, S_{23+}, S_{31+}\}$ iff $\left[ s< \frac{\sqrt 3}2 r, d_1 = \sqrt{\frac{r^2}4 + s^2\left(\frac{4r^2}{3r^2-4s^2}\right)^2},  d_3  = \frac{\sqrt{25r^4 - 24r^2 s^2 +16s^4}}{3r^2-4s^2}\right]$.}
\item{$X=\{ S_{12-},S_{23+}, S_{31+}\}$ iff $\left[ d_1> \sqrt{\frac{r^2}4 + s^2}, d_3=M\right]$.}
\item{$X=\{ S_{12+},S_{23-},S_{31-}\}$ iff $\left[s>\frac{\sqrt 3}2 r, d_1>\sqrt{\frac{r^2}4 + s^2 \left(\frac{4s^2+5r^2}{4s^2-3r^2}\right)^2}, d_3=d_3^*\right]$.}
\item{$X=\{ S_{12+}, S_{23+}, S_{31+}\}$ iff $\left[\sqrt{\frac{r^2}4 + \frac{s^2}9}<d_1<\sqrt{\frac{r^2}4 + s^2}, d_3=R \right]$, or $\left[ s=\frac{\sqrt 3 r}2. d_1>r, d_3=d_1 \right]$,\\
or $\left[ s<\frac{\sqrt 3}2 r, \sqrt{ \frac{r^2} 4 + s^2}\le d_1 <\sqrt{\frac{r^2}4 + s^2 \left(\frac{4r^2}{3r^2-4s^2}\right)^2}, d_3 = R\right]$,\\ 
or $\left[ s>\frac{\sqrt 3}2 r, \sqrt{ \frac{r^2} 4 + s^2}\le d_1 <\sqrt{\frac{r^2}4 + s^2 \left(\frac{4s^2+5r^2}{4s^2-3r^2}\right)^2},  d_3= R\right]$.}
\item{$X=S_{12\pm}$ iff $\left[ s<\frac{\sqrt 3}2 r,  d_1>\sqrt{\frac{r^2}4 + s^2\left(\frac{4r^2}{3r^2-4s^2}\right)^2}, d_3 = \sqrt{d_1^2 - \frac {r^2}4 - s^2}\right]$.}
\item{$X=\{S_{23-},S_{31-}\}$ iff $\left[ s>\frac{\sqrt 3}2 r,  d_1> \sqrt{\frac{r^2}4 +s^2\left(\frac{4s^2+5r^2}{4s^2-3r^2} \right)^2},  d_3^*< d_3< \sqrt{d_1^2 - \frac{r^2}4 - s^2}\right]$.}
\item{$X=\{S_{23+}, S_{31+}\}$ iff $\left[ \frac r 2<d_1< \sqrt{\frac{r^2}4 +\frac{s^2}9},  s-\sqrt{d_1^2-\frac{r^2}4 }<d_3<s+ \sqrt{d_1^2 - \frac{r^2}4 }\right]$, \\
or $\left[ \sqrt{\frac{r^2}4 +\frac{s^2}9}<d_1\le \sqrt{\frac{r^2}4 + s^2},  R<d_3<s+ \sqrt{d_1^2 - \frac{r^2}4 }\right]$, 
or $\left[ s=\frac{\sqrt 3 r}2, d_1>r, R<d_3<\sqrt{d_1^2 + 2r^2}\right]$,\\ 
or $\left[ s<\frac{\sqrt 3}2 r,  \sqrt{ \frac{r^2} 4 + s^2}< d_1 <\sqrt{\frac{r^2}4 + s^2 \left(\frac{4r^2}{3r^2-4s^2}\right)^2}, R<d_3<M\right]$, \\
or $\left[ s>\frac{\sqrt 3}2 r, d_1 > \sqrt{\frac{r^2}4 + s^2}, \max(R, \sqrt{d_1^2-\frac{r^2}4-s^2})<d_3<M\right]$.}
\end{itemize} 
\end{corollary}
Note that the cases in Corollary \ref{le:isoNumber} are all the nonunique cases for the source $X$, when $\norm{Z_1-Z_3}=\norm{Z_2-Z_3}$ and $d_1 = d_2$.

\section{Appendix}                                                                                                                                                                                                                                                                                                                                                                                                                                                                                                                                                                                                                                                                                                                                                                                                                                                                                                                                                                                                                                                                                                                                                                                                                                                                                                                                                                                                                                                                                                                                                                                                                                                                                                                                                                                                                                                                                                                                                                                                                    
We list the solutions in dictionary order for $d_1=d_2$ and $d_3$ in each cases for the isosceles measurement triangle with $ |Z_1 Z_3|=|Z_2 Z_3| $.
The number of solutions $|X|$ are also listed in the following tables. 
These results are the reorganization of  Theorems \ref{le:equid12}  and  \ref{le:isoscelesd12}. 

\newpage
\subsection{ Equilateral triangle case with $s=\frac{\sqrt 3}2 r$}\label{sec:equi}
This table is the reorganization of the table in Theorem \ref{le:equid12}.

\begin{tabular}{|c|c||c|c|}
\hline
$ d_1=d_2$                                        &$d_3$                                                                                                                                                        &$X$                                                        &$|X|$\\
\hline \hline
$\left(0,\frac r 2\right]$                      &$\left(0, \frac r{\sqrt 3}\right]$                                                                                                                    &$Y_0$                                                    &1\\
\hline
                                                           &$\left[ \frac r{\sqrt 3},\sqrt 3 r \right]$                                                                                                         &$N_3$                                                   &1\\
\hline                                                           
                                                           &$\left[ \sqrt 3 r, \infty \right)$                                                                                                                      &$Y_3$                                                   &1\\  
\hline
$\left(\frac r 2, \frac r{\sqrt 3}\right]$ &$\left(0, \frac r{\sqrt 3}\right]$                                                                                                                    &$Y_0$                                                   &1\\                                                        
\hline
                                                           &$\left[\frac r{\sqrt 3}, \frac{\sqrt 3}2 r- \sqrt{ d_1^2 - \frac{r^2}4 }\right]$                                                 &$N_3$                                                   &1\\
\hline
                                                           &$\left(\frac{\sqrt 3}2 r - \sqrt{ d_1^2 - \frac{r^2}4}, \frac{\sqrt 3}2 r + \sqrt{ d_1^2 - \frac{r^2}4 }\right)$ &$\{S_{23+}, S_{31+}\}$                   &2\\
\hline
                                                           &$\left[ \frac{\sqrt 3}2 r + \sqrt{ d_1^2 - \frac{r^2}4},   \sqrt 3 r \right]$                                                       &$N_3$                                                   &1\\
\hline                                                           
                                                           &$\left[ \sqrt 3 r, \infty \right)$                                                                                                                       &$Y_3$                                                   &1\\                                                             
\hline
$\left(\frac r{\sqrt 3}, r\right]$             &$\left(0, d_1\right)$                                                                                                                                     &$S_{12+}$                                           &1\\
\hline    
                                                           &$d_1$                                                                                                                                                          &$\{ S_{12+}, S_{23+}, S_{31+} \}$ &3\\                                                
\hline
                                                           &$\left( d_1, \frac{\sqrt 3}2 r+ \sqrt{ d_1^2 - \frac{r^2}4} \right)$                                                                &$\{S_{23+}, S_{31+}\}$                    &2 \\
\hline
                                                           &$\left[ \frac{\sqrt 3}2 r + \sqrt{ d_1^2 - \frac{r^2}4},   \sqrt 3 r \right]$                                                       &$N_3$                                                   &1\\
\hline                                                           
                                                           &$\left[ \sqrt 3 r, \infty \right)$                                                                                                                       &$Y_3$                                                   &1\\ 
\hline
$\left( r,\infty\right)$                           &$\left(0, d_1\right)$                                                                                                                                      &$S_{12+}$                                           &1\\
\hline    
                                                           &$d_1$                                                                                                                                                          &$\{ S_{12+}, S_{23+}, S_{31+} \}$ &3\\ 
\hline
                                                           &$\left( d_1,  \sqrt{d_1^2 + 2r^2}   \right)$                                                                                                   &$\{ S_{23+}, S_{31+} \}$                   &2\\
\hline
                                                          &$ \sqrt{d_1^2 + 2r^2}$                                                                                                                                 &$\{ S_{23+}, S_{31+}, S_{12-} \}$   &3\\
\hline 
                                                          &$\left(  \sqrt{d_1^2 + 2r^2},  \infty\right) $                                                                                                   &$S_{12-}$                                              &1\\                                                                  
\hline                                                                                                                                                                                                                                                                                                                                
\end{tabular}

\newpage
\subsection{ Flat isoscele trianlge case with $s<\frac{\sqrt 3}2 r$}
This table is the reorganization of the table in Theorem \ref{le:isoscelesd12}.

\begin{tabular}{|c|c|c||c|c|}
\hline
$ d_1=d_2$                                                                                                                                    &$d_3$   											&$X$                                                      &$|X|$\\
\hline \hline
$\left(0,\frac r 2 \right]$                                                                                                                  &$\left(0, \frac{2s} 3\right]$  								&$Y_0$                                                  &1\\
\hline
                                                                                       		                                             &$\left[\frac {2s}3 , 2s \right]$ 								&$N_3$                                                  &1\\
\hline
                                                                                       		                                             &$\left[ 2s, \infty \right)$ 									&$Y_3$                                                  &1\\
\hline
$\left(\frac r 2, \sqrt{ \frac{r^2} 4 + \frac {s^2} 9}\right]$                                                               &$\left(0, \frac{2s} 3\right]$  								&$Y_0$                                                  &1\\
\hline
                                                                                       		                                             &$\left[\frac {2s}3, s- \sqrt{ d_1^2 - \frac{r^2}4 }\right]$  				&$N_3$                                                  &1\\
\hline
                                                                                       		                                             &$\left(s - \sqrt{ d_1^2 - \frac{r^2}4}, s + \sqrt{ d_1^2 - \frac{r^2}4 }\right)$       &$\{S_{23+}, S_{31+}\}$                  &2\\
\hline
                                                                                       		                                             &$\left[ s+ \sqrt{ d_1^2 - \frac{r^2}4}, 2s \right]$ 					&$N_3$                                                 &1\\
\hline
                                  		                                                                                                  &$\left[ 2s, \infty \right)$ 									&$Y_3$                                                  &1\\
\hline
$\left(\sqrt{ \frac{r^2} 4 + \frac {s^2} 9},\sqrt{ \frac{r^2} 4 + s^2}\right]$                                     &$\left(0, R \right)$                                                                                                &$S_{12+}$                                         &1\\
\hline
                                                                                                                                                       &$R$                                                                                                                      &$\{S_{12+}, S_{23+}, S_{31+}\}$ &3\\
\hline
                                                                                                                                                       &$\left(R, s + \sqrt{d_1^2 - \frac {r^2} 4}\right)$                                                   &$\{S_{23+},S_{31+}\}$                    &2\\  
\hline
                                                                                       		                                            &$\left[ s+ \sqrt{ d_1^2 - \frac{r^2}4},   2s \right]$ 					&$N_3$                                                  &1\\
\hline
                                                                                       		                                            &$\left[ 2s, \infty \right)$ 									&$Y_3$                                                  &1\\
\hline                                                                                                                                                                                                                                                                                                                                
$\left(\sqrt{\frac{r^2}4+s^2},\sqrt{\frac{r^2}4+s^2\left(\frac{4r^2}{3r^2-4s^2}\right)^2}\right)$  &$\left(0,R\right)$                                                                                                   &$ S_{12+}$                                         &1\\ 
\hline
                                                                                                                                                      &$R$                                                                                                                       &$S_{123}^+$    &3\\
\hline
                                                                                                                                                      &$\left(R,M\right)$                                                                                                  &$\{S_{23+},S_{31+}\}$                    &2\\                                                        
\hline
                                                                                                                                                      &$M$      						           			           &$\{ S_{23+}, S_{31+}, S_{12-} \}$ &3\\
\hline 
                                           							                                &$\left( M,\infty \right) $  	                                                                                      &$S_{12-}$                                           &1\\ 
\hline
$\sqrt{\frac{r^2}4 + s^2 \left(\frac{4r^2}{3r^2-4s^2}\right)^2}$                                                   &$\left(0,\sqrt{d_1^2 - \frac{r^2}4 + s^2}\right)$                                                     &$S_{12+}$                                          &1\\    
\hline
                                                                                                                                                     &$\sqrt{d_1^2 - \frac{r^2}4 + s^2}$                                                                         &$\{S_{12\pm},S_{23+},S_{31+}\}$&4\\
\hline 
                                           							                                &$\left(\sqrt{d_1^2 - \frac{r^2}4 + s^2},\infty\right) $  	                                           &$S_{12-}$                                           &1\\    
\hline
$\left(\sqrt{\frac{r^2}4 + s^2 \left(\frac{4r^2}{3r^2-4s^2}\right)^2},\infty \right)$                        &$\left(0,\sqrt{d_1^2 - \frac{r^2}4 + s^2}\right)$                                                     &$S_{12+}$                                           &1\\                                                       
\hline 
                                                                                                                                                     &$\sqrt{d_1^2 - \frac{r^2}4 + s^2}$                                                                         &$S_{12\pm}$                                      &2\\
\hline
                                           							                                &$\left(\sqrt{d_1^2 - \frac{r^2}4 + s^2},\infty\right) $  	                                           &$S_{12-}$                                            &1\\                                                         
\hline
\end{tabular}

\newpage
\subsection{ Shapr isosceles triangle case with $s>\frac{\sqrt 3}2 r$}
This table is also the reorganization of the table in Theorem \ref{le:isoscelesd12}. \newline
\begin{tabular}{|c|c|c||c|c|}
\hline
$ d_1=d_2$                                                                                                                                                   &$d_3$   											&$X$                                                        &$|X|$\\
\hline \hline
$\left(0,\frac r 2 \right]$                                                                                                                                 &$\left(0, \frac{2s} 3\right]$  								&$Y_0$                                                    &1\\
\hline
                                                                                       		                                                           &$\left[\frac {2s}3 , 2s \right]$ 							&$N_3$                                                    &1\\
\hline
                                                                                       		                                                           &$\left[ 2s, \infty \right)$ 								&$Y_3$                                                     &1\\
\hline
$\left(\frac r 2, \sqrt{ \frac{r^2} 4 + \frac {s^2} 9}\right]$                                                                             &$\left(0, \frac{2s} 3\right]$  								&$Y_0$                                                     &1\\
\hline
                                                                                       		                                                           &$\left[\frac {2s}3, s- \sqrt{ d_1^2 - \frac{r^2}4 }\right]$  				&$N_3$                                                     &1\\
\hline
                                                                                       		                                                           &$\left(s - \sqrt{ d_1^2 - \frac{r^2}4}, s + \sqrt{ d_1^2 - \frac{r^2}4 }\right)$    &$\{S_{23+}, S_{31+}\}$                     &2\\
\hline
                                                                                       		                                                           &$\left[ s+ \sqrt{ d_1^2 - \frac{r^2}4}, 2s \right]$ 					&$N_3$                                                     &1\\
\hline
                                  		                                                                                                                &$\left[ 2s, \infty \right)$ 								&$Y_3$                                                      &1\\
\hline
$\left(\sqrt{ \frac{r^2} 4 + \frac {s^2} 9},\sqrt{ \frac{r^2} 4 + s^2}\right]$                                                   &$\left(0, R \right)$                                                                                             &$S_{12+}$                                             &1\\
\hline
                                                                                                                                                                     &$R$                                                                                                                   &$\{S_{12+}, S_{23+}, S_{31+}\}$     &3\\
\hline
                                                                                                                                                                     &$\left(R, s + \sqrt{d_1^2 - \frac {r^2} 4}\right)$                                                &$\{S_{23+},S_{31+}\}$                        &2\\  
\hline
                                                                                       		                                                          &$\left[ s+ \sqrt{ d_1^2 - \frac{r^2}4},   2s \right]$ 					&$N_3$                                                       &1\\
\hline
                                                                                       		                                                          &$\left[ 2s, \infty \right)$ 								&$Y_3$                                                        &1\\
\hline
$\left(\sqrt{\frac{r^2}4+s^2}, \sqrt{\frac{r^2}4 + s^2\left(\frac{4s^2+5r^2}{4s^2-3r^2}\right)^2}\right)$    &$\left(0, R\right)$                                                                                              &$S_{12+}$                                               &1\\ 
\hline
                                                                                                                                                                     &$R$                                                                                                                   &$S_{123}^+$         &3\\
\hline
                                                                                                                                                                     &$\left(R,M\right)$                                                                                              &$\{S_{23+},S_{31+}\}$                         &2\\
\hline
                                                                                                                                                                     &$M$      						           			           &$\{ S_{23+}, S_{31+}, S_{12-} \}$      &3\\
\hline 
                                                                 							                          &$\left( M,\infty\right) $  	                                                                           &$S_{12-}$                                                &1\\  
\hline
$\sqrt{\frac{r^2}4 + s^2 \left(\frac{4s^2+5r^2}{4s^2-3r^2}\right)^2}$                                                         &$\left(0, \sqrt{d_1^2 - \frac{r^2}4 - s^2} \right)$                                               &$S_{12+}$                                               &1\\                                                                    
\hline                                                      
                                                                                                                                                                     &$\sqrt{d_1^2 - \frac{r^2}4 - s^2}$                                                                     &$\{S_{12+},S_{23\pm},S_{31\pm}\}$ &5 \\  
\hline
                                                                                                                                                                     &$\left( \sqrt{d_1^2 - \frac{r^2}4 - s^2},M\right)$                                               &$\{S_{23+},S_{31+}\}$                         &2\\     
\hline
                                                                                                                                                                     &$M$      						           			           &$\{ S_{23+}, S_{31+}, S_{12-} \}$      &3\\
\hline 
                                                                 							                          &$\left( M,\infty\right) $  	                                                                           &$S_{12-}$                                                &1\\ 
\hline
$\left(\sqrt{\frac{r^2}4 + s^2 \left(\frac{4s^2+5r^2}{4s^2-3r^2}\right)^2}, \infty\right)$                              &$\left(0,d_3^*\right)$                                                                                        &$ S_{12+}$                                              &1\\
\hline
                                             			                                                                                          &$d_3^*$ 					                      			           &$\{ S_{12+}, S_{23-}, S_{31-} \}$       &3\\
\hline
							                                                                                          &$\left( d_3^*,    \sqrt{d_1^2 -\frac{r^2}4 - s^2} \right)$                                     &$\{S_{23-}, S_{31-}\}$                          &2\\              
\hline
                                                                                                                                                                    &$\sqrt{d_1^2 - \frac{r^2}4 - s^2}$                                                                      &$\{ S_{23\pm}, S_{31\pm} \}$             &4\\
\hline
                                                                                                                                                                     &$\left( \sqrt{d_1^2 - \frac{r^2}4 - s^2},M\right)$                                               &$\{S_{23+},S_{31+}\}$                         &2\\
\hline
                                                                                                                                                                     &$M$      						           			           &$\{ S_{23+}, S_{31+}, S_{12-} \}$      &3\\
\hline 
                                                                 							                          &$\left( M,\infty\right) $  	                                                                           &$S_{12-}$                                                &1\\
\hline
\end{tabular}

\centerline{{\bf Acknowledgements}}
This paper was supported by NRF (National Research Foundation of Korea) grant funded by Korea government (Ministry of Science and ICT:RS-2023-00242308). 

\centerline{{\bf Conflicts of interest}}
The author declares that the publication of this paper has no conflict of interest.


\end{document}